\documentclass[journal]{IEEEtran}
\ifCLASSOPTIONcompsoc
\else
\fi
\ifCLASSINFOpdf
\else
\fi
\usepackage{graphicx}
\usepackage{amssymb}
\usepackage{amsmath}
\usepackage{bm}
\usepackage[Algorithm, boxed]{algorithm}
\usepackage{algorithmic}
\usepackage[bookmarks=false]{hyperref}
\usepackage{breakurl}
\usepackage{multirow}
\newtheorem{theorem}{Theorem}

\newtheorem{definition}{Definition}
\newtheorem{lemma}{Lemma}
\newtheorem{remark}{Remark}
\newtheorem{example}{Example}

\begin{document}
%
\title{What are the Differences between Bayesian Classifiers and Mutual-Information Classifiers?}
%
%
%
%

\author{Bao-Gang Hu,~\IEEEmembership{Senior Member,~IEEE}
\IEEEcompsocitemizethanks{\IEEEcompsocthanksitem Bao-Gang Hu is with NLPR/LIAMA,
Institute of Automation, Chinese Academy of Sciences, Beijing 100190, China.\protect\\
E-mail: hubg@nlpr.ia.ac.cn}
\thanks{}}

\IEEEcompsoctitleabstractindextext{%
\begin{abstract}
In this study, both Bayesian classifiers and mutual-information classifiers are
examined for binary classifications
with or without a reject option. The general decision rules  
in terms of distinctions on error types and reject types are derived for Bayesian 
classifiers. A formal analysis is conducted to reveal the parameter 
redundancy of cost terms when abstaining classifications are enforced.
The redundancy implies an intrinsic 
problem of ``{\it non-consistency}'' for interpreting cost terms.
If no data is given to the cost terms,
we demonstrate the weakness of Bayesian classifiers in class-imbalanced classifications.
On the contrary, mutual-information classifiers are able to provide an objective solution from 
the given data, which shows a reasonable balance among error types and reject types. 
Numerical examples of using two types of classifiers are given 
for confirming the theoretical differences, including the extremely-class-imbalanced cases. 
Finally, we briefly summarize the Bayesian classifiers 
and mutual-information classifiers in terms of their application 
advantages, respectively.

\end{abstract}

}

\maketitle

\IEEEdisplaynotcompsoctitleabstractindextext

%
\IEEEpeerreviewmaketitle

\section{Introduction}
\label{sec:introduction}
The Bayesian principle provides a powerful and formal means of 
dealing with statistical inference in data processing, such as classifications \cite{kulkarni1998}. 
If classifiers are designed based on this principle, they are 
called ``{\it Bayesian classifiers}'' in this work. The learning targets 
for Bayesian classifiers are either the minimum error or the lowest cost. 
It was recognized that Chow \cite{chow1957}\cite{chow1970} was ``{\it among 
the earliest to use Bayesian decision theory for pattern recognition}'' 
\cite{duda2001}. His pioneering work is so enlightening that its idea of optimal 
tradeoff between error and reject still sheds a bright light for us to 
deep our understanding to the subject, as well as to explore its applications 
widely in this information-explosion era. In recent years, cost sensitive learning 
and class-imbalanced learning have received much attentions in various applications [12-18]. 
For classifications of imbalanced, or skewed, datasets, ``{\it the ratio of the small 
to the large classes can be drastic such as 1 to 100, 1 to 1,000, or 1 to 10,000 
(and sometimes even more)}'' \cite{chawla2004}. It was pointed out by Yang and Wu 
\cite{yang2006} that dealing with imbalanced and cost-sensitive data is among the 
ten most challenging problems in the study of data mining. In fact, the related subjects 
are not a new challenge but a more crucial concern than before for increasing needs of 
searching useful information from massive data. Binary classifications will be a 
basic problem in such application background. Classifications based on cost terms 
for the tradeoff of error types is a conventional subject in medical diagnosis. 
Misclassification from ``{\it type I error}'' (or ``{\it false positive}'') or 
from ``{\it type II error}'' (or ``{\it false negative}'') is significantly different 
in the context of medical practices. In other domains of applications, one also needs 
to discern error types for attaining reasonable results in classifications. Among all 
these investigations, cost terms, which is usually specified by users from a cost 
matrix, play a key role in class-imbalanced learning [11-14]\cite{breiman1984}\cite{santos-rodr2009}\cite{fu2010}.

In binary classifications with a reject option, Bayesian classifiers require a cost matrix 
with six cost terms as the given data. Different from the prior to the 
probabilities of classes, this requirement can be another source of subjectivity that disqualifies
Bayesian classifiers as an objective approach of induction \cite{berger2006}. 
If an objectivity aspect is enforced for classifications with a reject option, 
a difficulty does exist for Bayesian classifiers that assign cost terms objectively. 
The cost terms for error types may be given from an application background, but are generally 
unknown for reject types. In binary classifications, Chow \cite{chow1970} and early 
researchers \cite{ha1997}\cite{golfarelli1997}\cite{baram1998} usually assumed no 
distinctions among errors and among rejects. The later study in \cite{santos-perira2005} 
considered different costs for correct classification and miscalssifications, but not for rejects. 
The more general settings for distinguishing error types and reject types were reported 
in \cite{ferri2004}\cite{friedel2006}\cite{pietraszek2007}. To overcome the problems 
of presetting cost terms manually, Pietraszek \cite{pietraszek2007} proposed two learning 
models, namely, ``{\it bounded-abstention}'' and ``{\it bounded-improvement}'', 
and Grall-Ma\"{e}s and Beauseroy \cite{grall-maes2009} applied a strategy of adding 
performance constraints for class-selective rejection. If constraints either 
on total reject or on total error, they may result in no distinctions between their 
associated cost terms. Up to now, it seems that no study has  been reported for the objective 
design of Bayesian classifiers by distinguishing error types and reject types at the same time.

Several investigations are reported by following Chow's rule on classifier designs 
with a reject option [21-30]. In addition to a kind of ``{\it ambiguity reject}'' studied by 
Chow, the other kind of ``{\it distance reject}'' was also considered in \cite{dubuisson1993}. 
Ambiguity reject is made to a pattern located in an ambiguous region between/among classes. 
Distance reject represents a pattern far away from the means of any class and is 
conventionally called an ``{\it outlier}'' in statistics \cite{duda2001}. Ha \cite{ha1997} 
proposed another important kind of reject, called ``{\it class-selective reject}'', 
which defines a subset of classes. This scheme is more suitable to multiple-class 
classifications. For example, in three-class problems, Ha's classifiers will output 
the predictions including ``{\it ambiguity reject between Class 1 and 2}'', 
``{\it ambiguity reject among Class 1, 2 and 3}'', and the other rejects from class 
combinations. Multiple rejects with such distinctions will be more informative than 
a single ``{\it ambiguity reject}''. Among all these investigations, the Bayesian 
principle is applied again for their design guideline of classifiers.

While the Bayesian inference principle is widely applied in classifications, another 
principle based on the mutual information concept is rarely adopted for designing classifiers. 
Mutual information is one of the important definitions in entropy theory \cite{cover2006}. 
Entropy is considered as a measure of uncertainty within random variables, and mutual 
information describes the relative entropy between two random variables \cite{mackay2003}. 
If classifiers seek to maximize the relative entropy 
for their learning target, we refer them to ``{\it mutual-information classifiers}''. 
It seems that Quinlan \cite{quinlan1986} was among the earliest to apply the concept of 
mutual information (but called ``{\it information gain}'' in his famous ID3 algorithm) 
in constructing the decision tree. Kv\aa{}lseth \cite{kvalseth1987} and Wickens \cite{wichens1989} 
introduced the definition of normalized mutual information (NMI) for assessing a contingency 
table, which laid down the foundation on the relationship between a confusion matrix and mutual 
information. Being pioneers in using an information-based criterion for classifier evaluations, 
Kononenko and Bratko \cite{kononenko1991} suggested the term  ``{\it information score}'' 
which was equivalent to the definition of mutual information. A research team leaded by 
Principe \cite{principe2000} proposed a general framework, called ``{\it Information Theoretic 
Learning} ({ITL})'', for designing various learning machines, in which they suggested that mutual 
information, or other information theoretic criteria, can be set as an objective function 
in classifier learning. Mackay [\cite{mackay2003}, 
page 533] once showed numerical examples for several given confusion matrices, and he suggested 
to apply mutual information for ranking the classifier examples. Wang and Hu \cite{wang2008} 
derived the nonlinear relations between mutual information and the conventional performance 
measures, such as accuracy, precision, recall and F1 measure for binary classifications. 
In \cite{hu2008}, a general formula for normalized mutual information was established with 
respect to the confusion matrix for multiple-class classifications with/without a reject option, 
and the advantages and limitations of mutual-information classifiers were discussed. However, 
no systematic investigation is reported for a theoretical comparison between Bayesian classifiers 
and mutual-information classifiers in the literature.

This work focuses on exploring the theoretical differences between Bayesian classifiers and 
mutual-information classifiers in classifications for the settings with/without a reject option. 
In particular, this paper derives much from and consequently extends to Chow's work by 
distinguishing error types and reject types. To achieve analytical tractability without losing the 
generality, a strategy of adopting the simplest yet most meaningful assumptions to classification 
problems is pursued for investigations. The following assumptions are given in the same way as 
those in the closed-form studies of Bayesian classifiers by Chow \cite{chow1970} and Duda, et al 
\cite{duda2001}:
\begin{enumerate}
 \item[A1.] Classifications are made for two categories (or classes) over the feature variables.
 \item[A2.] All probability distributions of feature variables are exactly known. 
\end{enumerate}

One may argue that the assumptions above are extremely restricted to offer practical generality 
in solving real-world problems. In fact, the power of Bayesian classifiers does not stay 
within their exact solutions to the theoretical problems, but appear from their 
generic inference principle in guiding 
real applications, even in the extreme approximations to the theory. We fully recognize that 
the assumption of complete knowledge on the relevant probability distributions may be never the 
cases in real-world problems \cite{santos-perira2005}\cite{fukunaga1990}. The closed-form solutions 
of Bayesian classifiers on binary classifications in 
\cite{chow1970}\cite{duda2001} have demonstrated the useful design guidelines that are applicable 
to multiple classes \cite{ha1997}. The author believes that the analysis based 
on the assumptions above will provide sufficient information for revealing the theoretical 
differences between Bayesian classifiers and mutual-information classifiers, 
while the intended simplifications will benefit 
readers to reach a better, or deeper, understanding to the advantages and limitations of each type 
of classifiers. 

The contributions of this work are twofold. First, the analytical formulas for Bayesian classifiers 
and mutual-information classifiers are derived to include the general cases 
with distinctions among error 
types and reject types for cost sensitive learning in classifications. Second, comparisons are 
conducted between the two types of classifiers for revealing their similarities and differences. 
Specific efforts are made on a formal analysis of parameter redundancy to the cost terms for Bayesian 
classifiers when a reject option is applied. 
Section II presents a general decision rule of Bayesian classifiers with 
or without a reject option. Sections III provides the basic formulas for mutual-information classifiers. 
Section IV investigates the similarities and differences between two types of classifiers, and numerical 
examples are given to highlight the distinct features in their applications. The question presented in the title 
of the paper is concluded by a simple answer in Section V.

\section{Bayesian Classifiers with A Reject Option}
\label{sec:Bayesian_Classi_Reject}

\subsection{General Decision Rule for Bayesian Classifiers}
\label{subsec:GDR_BayesianClassi}
Let $\textbf{x}$ be a 
random pattern satisfying $\textbf{x} \in \textbf{X} \subset {R^d}$, which is 
in a $d$-dimensional feature space and will be classified. 
The true (or target) state $t$ of $\textbf{x}$ is within the finite set of 
two classes, $t \in {T} = \{t_{1}, t_{2}\}$, and the possible decision output $y=f(\textbf{x})$ is 
within three classes, $y \in {Y} = \{y_{1}, y_{2}, y_{3}\}$, where $f$ is a function for classifications
and $y_{3}$ represents 
a ``{\it reject}'' class. Let $p(t_{i})$ be the prior probability of class $t_{i}$ 
and $p(\textbf{x}|t_{i})$ be the conditional probability density function of $\textbf{x}$ given that it 
belongs to class $t_{i}$. The {\it posterior} probability $p(t_{i}|\textbf{x})$ is calculated 
through the Bayes formula \cite{duda2001}: 
\begin{equation}
 \label{equ:1}
p(t_{i}|\textbf{x}) = \dfrac{p(\textbf{x}|t_{i})p(t_{i})}{p(\textbf{x})},
\end{equation}
where $p(\textbf{x})$ represents the mixture density for normalizing the probability. Based on 
the posterior probability, the Bayesian rule assigns a pattern $\textbf{x}$ into the 
class that has the highest posterior probability. Chow \cite{chow1957}\cite{chow1970} 
first introduced the framework of the Bayesian decision theory into the study of pattern 
recognition and derived the best error-type trade-off formulas and the related optimal 
reject rule. The purpose of the reject rule is to minimize the total risk (or cost) in 
classifications. Suppose $\lambda_{ij}$ is a cost term for the true class of a pattern 
to be $t_i$, but decided as $y_j$. Then, the conditional risk for classifying a particular 
$\textbf{x}$ into $y_j$ is defined as:
\begin{equation}
 \label{equ:2}
 \begin{array}{r@{\quad}l}
& Risk(y_j|\textbf{x})=\sum \limits_{i=1}^{2}\lambda_{ij}p(t_i|\textbf{x}) =\sum 
\limits_{i=1}^{2}\lambda_{ij} \dfrac{p(\textbf{x}|t_i)p(t_i)}{p(\textbf{x})},\\
& j =1,2,3.
\end{array}
\end{equation}
Note that the definition of $\lambda_{ij}$ in this work is a bit different 
with that in \cite{duda2001}, so that $\lambda_{ij}$ will form a $2\times3$ matrix. 
Chow \cite{chow1970} assumed the initial constraints on $\lambda_{ij}$ from the intuition
in classifications: 
\begin{equation}
 \label{equ:3}
\lambda_{ik} > \lambda_{i3} > \lambda_{ii} \geq 0, \quad i \neq k, \quad i = 1,2, \quad k = 1,2.
\end{equation}
The constraints imply that a misclassification will suffer a higher cost than a rejection, and a rejection
will cost more than a correct classification. Relations about $\lambda_{ij}$ are 
the main concern in the study of cost-sensitive learning, and this issue will be 
addressed later in this work. The total risk for the decision output $y$ will be \cite{duda2001}:
\begin{equation}
 \label{equ:4}
Risk(y)=\int \limits_{V} 
\sum \limits_{j=1}^{3}\sum \limits_{i=1}^{2}\lambda_{ij}p(t_i|\textbf{x})p(\textbf{x})d\textbf{x},
\end{equation}
with integration over the entire observation space $V$.

\begin{definition}[Bayesian classifier]
If a classifier is determined from the minimization of its risk over all patterns:
\begin{equation}\tag{5a}
 \label{equ:5a}
y^*=arg \min \limits_{y}Risk(y),
\end{equation}
or in anther form on a given pattern $\textbf{x}$:
\begin{equation}\tag{5b}
 \label{equ:5b}
Decide \quad y_j \quad if \quad Risk(y_j|\textbf{x})=\min \limits_{i}Risk(y_i|\textbf{x})
\end{equation}
this classifier is called ``{\it Bayesian classifier}'', or ``{\it Chow's abstaining 
classifier}'' \cite{friedel2006}. The term of $Risk(y^*)$ is usually called 
``{\it Bayesian risk}'', or ``{\it Bayesian error}'' in the cases that zero-one 
cost terms ($\lambda_{11}=\lambda_{22}=0, \lambda_{12}=\lambda_{21}=1$) are used 
for no rejection classifications\cite{duda2001}.
\end{definition}

In \cite{chow1970}, a single threshold for a reject option was investigated. 
This setting was obtained for the assumption that cost terms are applied without 
distinction among the errors and among rejects. Following Chow's approach but with 
extension to the general cases to cost terms, one is able to derive the general decision 
rule on the rejection for Bayesian classifiers. 

\begin{theorem}
 The general decision rule for Bayesian classifiers are:
 \begin{equation}\tag{6a}
\label{equ:6a}
 \begin{array}{r@{\quad}l}
& Decide ~ y_1 ~ if ~ \dfrac{p(\textbf{x}|t_1)p(t_1)}{p(\textbf{x}|t_2)p(t_2)}>\delta_1,\\
& No~ rejection: ~ \delta_1=\dfrac{\lambda_{21} - \lambda_{22}}{\lambda_{12} - \lambda_{11}},\\
& Rejection: ~ \delta_1=\dfrac{\lambda_{21} - \lambda_{23}}{\lambda_{13} - \lambda_{11}},
\end{array} 
\end{equation}

 \begin{equation}\tag{6b}
\label{equ:6b}
 \begin{array}{r@{\quad}l}
& Decide ~ y_2 ~ if \quad \dfrac{p(\textbf{x}|t_1)p(t_1)}{p(\textbf{x}|t_2)p(t_2)} \leq \delta_2,\\
& No ~ rejection: ~ \delta_2=\dfrac{\lambda_{21} - \lambda_{22}}{\lambda_{12} - \lambda_{11}},\\
& Rejection: ~ \delta_2=\dfrac{\lambda_{23} - \lambda_{22}}{\lambda_{12} - \lambda_{13}},
\end{array} 
\end{equation}

 \begin{equation}\tag{6c}
\label{equ:6c}
 \begin{array}{r@{\quad}l}
& Decide ~ y_3 ~ if ~ \dfrac{T_{r2}}{1 - T_{r2}}=\dfrac{\lambda_{23} 
- \lambda_{22}}{\lambda_{12} - \lambda_{13}}\\
& < \dfrac{p(\textbf{x}|t_1)p(t_1)}{p(\textbf{x}|t_2)p(t_2)} \leq \dfrac{\lambda_{21} - 
\lambda_{23}}{\lambda_{13} - \lambda_{11}}=\dfrac{1 - T_{r1}}{T_{r1}},
\end{array} 
\end{equation}

 \begin{equation}\tag{6d}
\label{equ:6d}
 \begin{array}{r@{\quad}l}
& Subject ~ to ~ 0<\dfrac{\lambda_{23} - \lambda_{22}}{\lambda_{12} - 
\lambda_{13}} < \dfrac{\lambda_{21} - \lambda_{22}}{\lambda_{12} - \lambda_{11}}\\
& < \dfrac{\lambda_{21} - \lambda_{23}}{\lambda_{13} - \lambda_{11}}, \quad and
\end{array}
\end{equation}

 \begin{equation}\tag{6e}
\label{equ:6e}
 \begin{array}{r@{\quad}l}
& No ~ rejection: ~ T_{r1} = T_{r2} = 0.5,  \\
& Rejection: ~ 0<T_{r1} + T_{r2}\leq 1.
\end{array}
\end{equation}
Eq (\ref{equ:6c}) applies the definition of two thresholds 
(called ``{\it rejection thresholds}'' in \cite{chow1970}), $T_{r1}$ and $T_{r2}$. 
\end{theorem}

\begin{proof}
 See Appendix A.
\end{proof}

\begin{figure*}[!t]
\centering
\includegraphics[width=8.4in,height=4.4in, bb= -70 -0 700 410]{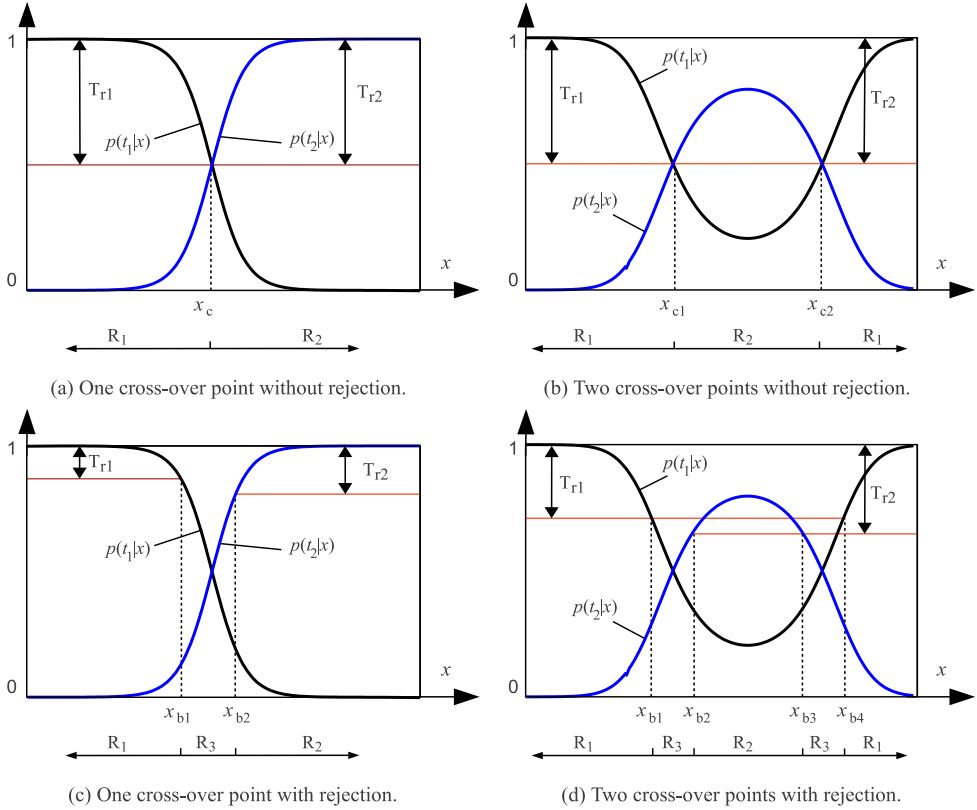}
\centering
\caption{Rejection scenarios from the plots of $p(t_i|x)$ for univariate Gaussian distributions.}
\label{fig:1}
\end{figure*}

Note that eq. (\ref{equ:6d}) suggests general constraints over $\lambda_{ij}$. 
The necessity for having such constraints is explained in Appendix A. 
A graphical interpretation to the two thresholds is illustrated in Fig. 1. Based on eq. (\ref{equ:6c}), 
the thresholds can be calculated from the following formulas: 
\setcounter{equation}{6}
\begin{equation}
 \label{equ:7}
\begin{array}{r@{\quad}l}
& T_{r1} = \dfrac{\lambda_{13} - \lambda_{11}}{\lambda_{13} - \lambda_{11} + \lambda_{21} - \lambda_{23}}, ~ and \\
  & T_{r2} = \dfrac{\lambda_{23} - \lambda_{22}}{\lambda_{12} - \lambda_{13} + \lambda_{23} - \lambda_{22}}.
\end{array}
\end{equation}
Eq. (\ref{equ:7}) describes general relations between thresholds and cost terms on 
binary classifications, which enables the classifiers to make the distinctions among 
errors and among rejects. Note that the special settings of Chow's rules \cite{chow1970} 
can be derived from eq. (\ref{equ:7}):
\begin{equation}
 \label{equ:8}
\lambda_{11} = \lambda_{22} = 0,~ \lambda_{12} = \lambda_{21} = 1, ~ \lambda_{13} = \lambda_{23} =  T_{r}.
\end{equation}
Another important relation in \cite{pietraszek2007} can also be obtained: 
\begin{equation}
 \label{equ:8}
\begin{array}{l@{\quad}}
\lambda_{11}=\lambda_{22}=0, \\
0 < \lambda_{r} = \lambda_{13} = \lambda_{23} < 
\dfrac{\lambda_{12}\lambda_{21}}{\lambda_{12}+\lambda_{21}} , \\
T_{r1}=\dfrac{\lambda_{r}}{\lambda_{21}} ~~~~and~~~~ T_{r2}=\dfrac{\lambda_{r}}{\lambda_{12}} .
\end{array}
\end{equation}
Pietraszek \cite{pietraszek2007} derived the rational region of $\lambda_{r}$ above through ROC curves. 
The error costs can be different but not for reject ones. Note that, 
however, the rejection thresholds will be different when $\lambda_{12} \neq \lambda_{21}$.    
For advanced applications, Vanderlooy, et al \cite{vanderlooy2009} generalized Chow's rules by distinguishing error 
types and reject types, and derived the relations between two 
''\emph{likelihood ratio thresholds}`` and cost terms. 
Their rules of missing the terms $\lambda_{11}$ and $\lambda_{22}$ are 
not theoretically general, yet sufficient for applications. They derived formulas only from the 
inequality constraints of $Risk(y_1|\textbf{x}) > Risk(y_3|\textbf{x})$ and $Risk(y_2|\textbf{x}) 
> Risk(y_3|\textbf{x})$, respectively. 
Up to now, it seems no one has reported the general constraints (\ref{equ:6d}) in the literature. 
Based on eq. (\ref{equ:6d}), one can derive the rational (\ref{equ:3}), rather than employing the intuition. 

By applying eq. (\ref{equ:1}) and the constraint $p(t_1|\textbf{x})+p(t_2|\textbf{x}) =1$, one can achieve 
the decision rules from eq. (6) with respect to the posterior probabilities and 
thresholds in a simple and better form for abstaining classifiers:
 \begin{equation}
\label{equ:9}
\begin{array}{r@{\quad}l}
Decide ~ y_1 ~ if ~ p(t_1|\textbf{x})>1 - T_{r1},~~~~~~\\
Decide ~ y_2 ~ if ~ p(t_2|\textbf{x}) \geq 1 - T_{r2},~~~~~~\\
Decide ~ y_3 ~ for ~ otherwise, ~~~~~~~~~~~~~\\
 Subject ~ to ~ 0 < T_{r1} + T_{r2} \leq 1. ~~~~~~
\end{array}
\end{equation}
In comparison with the decision rules of eq. (6), which are expressed in terms of
the likelihood ratio, eq. (\ref{equ:9}) together with Fig. 1 presents a better view for users to understand 
abstaining Bayesian classifiers. A plot of posterior probabilities show advantages over a
plot of the likelihood ratio (Figure 2.3 in \cite{duda2001}) for determining rejection
thresholds. Note that in Fig. 1 the plots are depicted on a one-dimensional variable
for Gaussian distributions of $X$. The simplification supports the suggestions by Duda, et
al, that one ``{\it should not obscure the central points illustrated in our simple 
example}'' \cite{duda2001}. Two sets of geometric points are shown for the 
plots. One set is called ``{\it cross-over points}'', denoted by $x_{ci}$, which are 
formed from two curves of $p(t_1|x)$ and $p(t_2|x)$. And the other is termed 
``{\it boundary points}'', denoted by $x_{bj}$. The boundary points partition 
classification regions for one-dimensional problems. For a ``{\it no rejection}'' 
case, the boundary points are controlled by the ratio of $(\lambda_{21} - 
\lambda_{22})/(\lambda_{12} - \lambda_{11})$. In abstaining classifications, 
those points are determined from two thresholds, respectively. For multiple dimension
problems, one can understand that both types of the points above become to be curves
or even hypersurfaces.  

With the exact knowledge of $p(t_i)$, $p(\textbf{x}|t_i)$, and $\lambda_{ij}$, one can 
calculate Bayesian risk from the following equation:
 \begin{equation}
\label{equ:10}
 \begin{array}{r@{}l}
& Risk(y^*) = \lambda_{11}CR_1 + \lambda_{12}E_1 + \lambda_{13}Rej_1 + \lambda_{22}CR_2 \\
& + \lambda_{21}E_2 + \lambda_{23}Rej_2\\
& = \lambda_{11}\int \limits_{R_1}p(t_1)p(\textbf{x}|t_1)d\textbf{x} + \lambda_{12}\int \limits_{R_2}p(t_1)p(\textbf{x}|t_1)d\textbf{x}\\
& + \lambda_{13}\int \limits_{R_3}p(t_1)p(\textbf{x}|t_1)d\textbf{x} +\lambda_{21}\int \limits_{R_1}p(t_2)p(\textbf{x}|t_2)d\textbf{x}\\
& + \lambda_{22}\int \limits_{R_2}p(t_2)p(\textbf{x}|t_2)d\textbf{x} + \lambda_{23}\int \limits_{R_3}p(t_2)p(\textbf{x}|t_2)d\textbf{x},
\end{array}
\end{equation}
where $CR_i$, $E_i$ and $Rej_i$ are the probabilities of ``{\it Correct Recognition}'', 
``{\it Error}'', and ``{\it Rejection}'' for the $i$th class in the classifications, 
respectively; and $R_1$ to $R_3$ are the classification regions of Class 1, Class 2 and 
the reject class, respectively. The general relations among $CR_i$, $E_i$ and $Rej_i$ for 
binary classifications are given by \cite{chow1970}:
 \begin{equation}
\label{equ:11}
 \begin{array}{r@{}l}
& CR_1 + CR_2 + E_1 + E_2 + Rej_1 + Rej_2 \quad \quad \quad \\
& \qquad \qquad \qquad \qquad \qquad= CR + E + Rej = 1,\\
& and \quad A = \dfrac{CR}{CR + E},
\end{array}
\end{equation}
where $CR$, $E$, and $Rej$ represent total correct recognition, total error and total 
reject rates, respectively; and $A$ is the accuracy rate of classifications.

\subsection{Parameter Redundancy Analysis of Cost Terms}
\label{subsec:2.4}
Bayesian classifiers present one of the general tools for cost sensitive learning. From 
this perspective, there exists a need for a systematic investigation into a parameter redundancy 
analysis of cost terms for Bayesian classifiers, which appears missing for a reject option. 
This section will attempt to develop a theoretical analysis of parameter redundancy for cost terms.

For Bayesian classifiers, when all cost terms are given along with the other 
relevant knowledge about classes, a unique set of solutions will be obtained. However, 
this phenomenon does not indicate that all cost terms will be independent for determining 
the final results of Bayesian classifiers. In the followings, a parameter dependency analysis 
is conducted because it  suggests a theoretical basis for a better understanding of relations 
among the cost terms and the outputs of Bayesian classifiers. Based on \cite{yang2008}\cite{hu2009}, 
we present the relevant definitions but derive a theorem from the functionals in 
eqs. (4) and (5) so that it holds generality for any distributions of features. 
Let a parameter vector be defined as $\mathbf{\theta} = \{\theta_1, 
\theta_2, \cdots, \theta_p\} \in \textbf{S}$, where $p$ is the total number of parameters in a model 
$f(\textbf{x}, \mathbf{\theta})$ and $\textbf{S}$ denotes the parameter space.

\begin{definition}[Parameter redundancy \cite{yang2008}]
A model $f(\textbf{x},\mathbf{\theta})$ is considered to be parameter redundant if it can be expressed 
in terms of a smaller sized parameter vector $\mathbf{\beta} = \{\beta_1, \beta_2, \cdots, 
\beta_q\} \in \mathbf{S}$, where $q < p$.
\end{definition}

\begin{definition}[Independent parameters]
A model $f(\textbf{x}, \mathbf{\beta})$ is said to be governed by independent parameters if it can be 
expressed in terms of the smallest size of parameter vector $\mathbf{\beta} = \{\beta_1, 
\beta_2, \cdots, \beta_m\} \in \textbf{S}$. Let $N_{IP}(\mathbf{\beta})$ denote the total 
number $(=m)$ of $\mathbf{\beta}$ for the model $f(\textbf{x}, \mathbf{\beta})$.
\end{definition}

\begin{definition}{\it (Function of parameters, parameter composition, input parameters, 
intermediate parameters)}:
Suppose three sets of parameter vectors are denoted by $\mathbf{\theta}= \{\theta_1, 
\theta_2, \cdots, \theta_p\} \in \textbf{S}_1$, $\mathbf{\gamma} = \{\gamma_1, \gamma_2, \cdots, 
\gamma_q\} \in \textbf{S}_2$, and $\mathbf{\eta}= \{\eta_1, \eta_2, \cdots, \eta_r\} \in \textbf{S}_3$.
If for a model there exists $f(\textbf{x}, 
\mathbf{\theta})=f(\textbf{x}, \varphi(\psi(\mathbf{\theta})))$ for $\varphi$: $\textbf{S}_1 \rightarrow 
\textbf{S}_2$ and $\psi$: $\textbf{S}_2 \rightarrow \textbf{S}_3$, we call $\varphi$ and $\psi$ to be functions of 
parameters, and $\varphi(\psi(\mathbf{\theta}))$ to be parameter composition, where $\theta_i$ 
are called input parameters for $f(\textbf{x}, \varphi(\psi(\mathbf{\theta})))$, $\gamma_j$ and $\eta_k$ 
are intermediate parameters.
\end{definition}

\begin{lemma}
 Suppose a model holds the relation $f(\textbf{x}, \mathbf{\theta})=f(\textbf{x}, \varphi(\psi(\mathbf{\theta})))$ 
for Definition 4. The total number of independent parameters of $\mathbf{\theta}$, denoted as 
$N_{IP}(f, \mathbf{\theta})$ for the model $f$ will be no more than $\min(p, q, r)$, or in a form of:
\begin{equation}
\label{equ:26}
 N_{IP}(f, \mathbf{\theta}) \leq \min (p, q, r)
\end{equation}
\end{lemma}

\begin{proof}
 Suppose $f(\textbf{x}, \mathbf{\theta} = \{\theta_1, \theta_2, \cdots, \theta_p\})$ without parameter 
composition, one can prove that $N_{IP}(f, \mathbf{\theta}) \leq \min(p)$. According to 
Definition 2, any increase of its size of $\mathbf{\theta}$ over $p$ will produce a
parameter redundancy in the model. Definition 3 indicates that the vector size $p$ will be 
an upper bound for $N_{IP}(f, \mathbf{\theta})$ in this situation. In the same principle, 
after parameter compositions are defined in Definition 4 for $f(\textbf{x}, \mathbf{\theta})=f(\textbf{x}, 
\varphi(\psi(\mathbf{\theta})))$, the lowest parameter size within $\mathbf{\theta}$, $\psi$ 
and $\varphi$, will be the upper bound of $f(\textbf{x}, \mathbf{\theta})$.   
\end{proof}

For Bayesian classifiers defined by eq. (\ref{equ:5a}), one can rewrite it in a form of:
\begin{equation}
 \label{equ:27}
y^* = arg \min Rsik(y, \{\mathbf{\theta_\lambda}, \mathbf{\theta_C} \}),
\end{equation}
where $\mathbf{\theta_\lambda} = (\lambda_{11}, \lambda_{12}, \lambda_{13}, \lambda_{21}, 
\lambda_{22}, \lambda_{23})$ and $\mathbf{\theta}_C = ({p(t_1), p(t_2), p(\textbf{x}|t_1), p(\textbf{x}|t_2)}$ 
in binary classifications, with $\mathbf{\theta_\lambda} \cap \mathbf{\theta_C} = \emptyset$ 
for their disjoint sets. Let $E$ (or $Rej$) be the total Bayesian error (or reject) in binary 
classifications:
\begin{equation}
 \label{equ:28}
 \begin{array}{r@{}l}
& E(y^*, \mathbf{\theta}) = E_1 + E_2 = \qquad \qquad \qquad \qquad \qquad\\
& \qquad \qquad \int \limits_{R_2}p(t_1)p(\textbf{x}|t_1)d\textbf{x} + \int \limits_{R_1}p(t_2)p(\textbf{x}|t_2)d\textbf{x},\\
& Rej(y^*, \mathbf{\theta}) = Rej_1 + Rej_2 = \\
& \qquad \qquad \int \limits_{R_3}p(t_1)p(\textbf{x}|t_1)d\textbf{x} + \int \limits_{R_3}p(t_2)p(\textbf{x}|t_2)d\textbf{x}.
\end{array}
\end{equation}

Based on eqs. (\ref{equ:7}) and (\ref{equ:11}), the total error (or reject) of Bayesian 
classifiers defined by eq. (\ref{equ:28}) shows a form of composition of parameters:
\begin{equation}
 \label{equ:29}
 \begin{array}{r@{=}l}
E(y^*, \{\mathbf{\theta_\lambda, \theta_C}\}) & E(y^*, \{\mathbf{x_b(T_r(\theta_\lambda)), \theta_C}\}),\\
Rej(y^*, \{\mathbf{\theta_\lambda, \theta_C}\}) & Rej(y^*, \{\mathbf{x_b(T_r(\theta_\lambda)), \theta_C}\})
\end{array}
\end{equation}
where $\mathbf{x_b}$ and $\mathbf{T_r}$ are two functions of the parameters.  
$\lambda_{ij}~ (i=1,2, j=1,2,3)$ are usually input parameters, but $T_{rk}~ (k=1,2)$ 
can serve as either intermediate parameters or input ones.

\begin{theorem}
 In abstaining binary classifications, the total number of independent parameters within 
the cost terms for defining Bayesian classifiers, $y^*$, should be at most two $(N_{IP}(y^*, 
\mathbf{\theta}) \leq 2)$. Therefore, applications of cost terms of $\mathbf{\theta_\lambda}= 
(\lambda_{11},\lambda_{12},\lambda_{13}, \lambda_{21},\lambda_{22},\lambda_{23})$ in the 
traditional cost sensitive learning will exhibit a parameter redundancy for calculating Bayesian 
$E(y^*)$ and $Rej(y^*)$ even after assuming $\lambda_{11} = \lambda_{22} = 0$, 
and $\lambda_{12} =1$ as the conventional way in classifications \cite{elkan2001}\cite{friedel2006}.
\end{theorem}

\begin{proof}
 Applying (\ref{equ:27}) and (\ref{equ:26}) in Lemma 1, one can have $N_{IP}(y^*, 
\mathbf{\theta}) \leq \min(p=6, q=2, r=4)=2$ for defining Bayesian classifiers from 
$\mathbf{\theta}$. However, when imposing three constraints on $\lambda_{11} = \lambda_{22} = 0$, 
and $\lambda_{12} =1$, $\mathbf{\theta}$ will provide three free parameters in the cost matrix in a form of:
\begin{equation}
 \label{equ:30}
\begin{array}{r@{=}l}
\lambda_{21} & \lambda_{21}\\
\lambda_{13} & \dfrac{T_{r1}(T_{r2}*\lambda_{21} + T_{r2} - \lambda_{21})}{T_{r1} + T_{r2} - 1}\\
\lambda_{23} & \dfrac{T_{r2}(T_{r1}*\lambda_{21} + T_{r1} - \lambda_{21})}{T_{r1} + T_{r2} - 1},
\end{array}
\end{equation}
which implies a parameter redundancy for calculating Bayesian $E(y^*)$ and $Rej(y^*)$.
\end{proof}

\begin{remark}
 Theorem 2 describes that Bayesian classifiers with a reject option will suffer a 
difficulty of uniquely interpreting cost terms. For example, one can even enforce 
the following two settings:
\begin{equation}
\left\{ \begin{array}{r@{}l}
& \lambda_{11} = 0, ~\lambda_{12} = 1, ~ 0 \leq \lambda_{13} \leq 1,\\ \nonumber
& \lambda_{21} = 1, ~\lambda_{22} = 0, ~ 0 \leq \lambda_{23} \leq 1, \nonumber
\end{array} \right.
\end{equation}
or
\begin{equation}
\left\{ \begin{array}{r@{}l}
& \lambda_{11} = 0, ~1 \leq \lambda_{12}, ~\lambda_{13} = 1,\\ \nonumber
& 1 \leq \lambda_{21}, ~\lambda_{22} = 0, ~\lambda_{23} = 1. \nonumber
\end{array} \right.
\end{equation}
for achieving the same Bayesian classifier, as well as their $E(y^*)$ and $Rej(y^*)$. 
However, the two sets of settings entail different meanings and do not show the
equivalent relations except through eq. (\ref{equ:7}). Hence, a confusion may be introduced 
when attempting to understand behaviors of error and reject rates with respects to different sets of cost 
terms. For this reason, cost terms may present an intrinsic problem for defining a 
generic form of settings in cost sensitive learning if a reject option is enforced.
\end{remark}

\begin{remark}
 While Theorem 2 only shows an estimation of upper-bound of $N_{IP}(y^*, \mathbf{\theta})$ 
for Bayesian classifiers with a reject option because of missing a closed-form solution of 
$E(y^*, \mathbf{\theta})$, one can prove on $N_{IP}(y^*, \mathbf{\theta})=1$ for Bayesian 
classifiers without rejection. A single independent parameter from the cost terms can be 
formed as $(\lambda_{12} - \lambda_{11})/(\lambda_{21} - \lambda_{22})$.
\end{remark}

\begin{remark}
 We suggest to apply independent parameters for the design and cost analysis of Bayesian 
classifiers. The total number of independent parameters of $N_{IP}(y^*, \mathbf{\theta})$ 
is changeable and dependent on the reject option of Bayesian classifiers. If rejection is not considered, 
we suggest $\mathbf{\theta}=(\lambda_{11}=\lambda_{22}=0, \lambda_{12}=1, \lambda_{21}>0)$ 
for the cost or error sensitivity analysis. A single independent cost parameter, $\lambda_{21}$, 
is capable of governing complete behaviors of error rate. For a reject option, we suggest 
$\mathbf{\theta}=(0 \leq T_{r1}, 0 \leq T_{r2}, ~and~ T_{r1}+T_{r1} \leq 1)$ for the cost, 
error, or reject sensitivity analysis, which will lead to a unique interpretation to the analysis.
\end{remark}

\subsection{Examples of Bayesian Classifiers on Univariate 
Gaussian Distributions}
\label{subsec:2.2}
This section will consider abstaining Bayesian classifiers on Gaussian distributions. 
As a preliminary study, a univariate feature in \cite{duda2001} is adopted for the 
reason of showing theoretical fundamentals as well as the closed-form solutions. Therefore, 
if the relevant knowledge of $p(t_i)$ and $p(x|t_i)$ is given, one can depict the plots 
of $p(t_i|x)$ from calculation of eq. (\ref{equ:1}) (Fig. 1). Moreover, when 
$\lambda_{ij}$ is known, the classification regions 
of $R_1$ to $R_3$ in terms of $x_{bj}$ will be fixed for Bayesian classifiers. After the 
regions $R_1$ to $R_3$, or $x_{bj}$, are determined, Bayesian risk will be obtained directly. 
One can see that these boundaries can be obtained from the known data of $\delta_i$ when 
solving an equality equation on (\ref{equ:6a}) or (\ref{equ:6b}): 
\begin{equation}
 \label{equ:13}
\dfrac{p(x=x_c|t_1)p(t_1)}{P(x=x_c|t_2)p(t_2)} = \delta_{i}
\end{equation}
The data of $\delta_i$ can be realized either from cost terms $\lambda_{ij}$, or from 
threshold $T_{ri}$ (see eq. (6)). By substituting the exact data of $p(t_i)$ and
$p(x|t_1)\sim N(\mu_i, \sigma_i)$ for Gaussian 
distributions, where $\mu_i$ and $\sigma_i$ represent the mean and standard deviation to the 
$i$th class, and the data of $\delta_i$ (say, for $\delta_1=(1-T_{r1})/T_{r1}$ from 
the given $T_{r1}$) into (\ref{equ:13}), one can obtain the closed-form solutions to 
the boundary points (say, for $x_{b1}$ and $x_{b4}$):
\begin{equation}\tag{19a}
\label{equ:14a}
x_{b1,4}=\dfrac{\mu_2\sigma^2_1 - \mu_1\sigma^{2}_{2}}{\sigma^{2}_{1} - \sigma^{2}_{2}} 
\mp \dfrac{\sigma_1 \sigma_2 \sqrt{\alpha}}{\sigma^{2}_{1} - \sigma^{2}_{2}}, ~ if ~ \sigma_1 \neq \sigma_2
\end{equation}

\begin{equation}\tag{19b}
\label{equ:14b}
x_{b1}=\dfrac{\mu_1 + \mu_2}{2} + \dfrac{\sigma^2}{\mu_2 - \mu_1}ln(\dfrac{p(t_1)}{p(t_2)} 
\dfrac{1}{\delta_1}), ~ if ~ \sigma_1 = \sigma_2 = \sigma
\end{equation}
where $\alpha$ is an intermediate variable defined by:
\begin{equation}\tag{19c}
\label{equ:14c}
\alpha = (\mu_1 - \mu_2)^2 - (2\sigma^{2}_{1} - 2\sigma^{2}_{2})ln(\dfrac{p(t_1)\sigma_2}{p(t_2)\sigma_1} 
\dfrac{1}{\delta_1}).
\end{equation}
Eq. (19) is also effective for Bayesian classifiers in the case of ``{\it no rejection}''. 
However, only cost terms, $\lambda_{ij} (i,j=1,2)$, will define the data of $\delta_1$. 
The general solution to abstaining classifiers has four boundary points by substituting 
two threshold $T_{r1}$ and $T_{r2}$, respectively. For the conditions shown in 
Fig. 1d, $T_{r1}$ will lead to $x_{b1}$ and $x_{b4}$, and $T_{r2}$ to $x_{b2}$ 
and $x_{b3}$, respectively. Eq. (\ref{equ:14a}) shows a general form for achieving 
two boundary points from one data point of $\delta_1$, and eq. (\ref{equ:14b}) is specific 
for reaching a single boundary point only when the standard deviations of two classes 
are the same. Substituting the other data of $\delta_2$ into eq. (19) will yield 
another pair of data $x_{b2}$ and $x_{b3}$, or a single one $x_{b2}$, in a similar 
form of eq. (19).

\begin{table*}
 \label{tab:1}
\renewcommand{\arraystretch}{1.3}
\centering
\caption{Rejection Settings for Bayesian Classifiers in univariate Gaussian Distributions 
\newline
$(x_{b1}<x_c<x_{b2} \quad or \quad x_{b1}<x_{c1}<x_{b2}<x_{b3}<x_{c2}<x_{b4})$}
\setlength{\tabcolsep}{3pt}
\centering
\begin{tabular}{|c|c|c|c|}
 \hline
Cross-over Point(s) & Rejection & Reject & \\
(Reference Figure) & Thresholds & region(s) & Remarks\\
\hline
 & $T_{r1} = 0.5$, $T_{r2} = 0.5$ & $\emptyset$ & No Rejection \\
 \cline{2-4}
 & $T_{r1} \geq 0.5$, $1-\max(p(t_2|x))<T_{r2}<0.5$ & $ \left[ x_{c1}, x_{b2} 、\right) $ and 
$ \left( x_{b3}, x_{c2} \right] $ & - \\
 \cline{2-4}
 Two & $T_{r1}<0.5$, $T_{r2} \geq 0.5$ & $ \left[ x_{b1}, x_{c1} \right) $ and $ \left( x_{c2}, x_{b4} \right] $ & - \\
 \cline{2-4} 
(Fig. 1d)  & $T_{r1}<0.5$, $1-\max(p(t_2|x))<T_{r2}<0.5$ & $ \left[ x_{b1}, x_{b2} \right) $ and $\left( x_{b3}, x_{b4} \right] $ & General Rejection \\
 \cline{2-4}
 & $T_{r1} < 0.5$, $T_{r2} < 1-\max(p(t_2|x))$ & $[x_{b1}, x_{b4}]$ & 
``{\it Class-1 and Reject-class}'' Classification\\
\cline{2-4}
& $T_{r1}=0$, $T_{r2} < 1$ & $(-\infty, x_{b2})$ and  $(x_{b3}, \infty)$ & 
``{\it Class-2 and Reject-class}'' Classification\\
\hline
 & $T_{r1}=0.5$, $T_{r2}=0.5$ & $\emptyset$ & No Rejection \\
\cline{2-4}
One & $T_{r1} \geq 0.5$, $T_{r2} < 0.5$ & $[x_{c}, x_{b2})$ & - \\
\cline{2-4}
(Fig. 1c) & $T_{r1} < 0.5$, $T_{r2} \geq 0.5$ & $[x_{b1}, x_{c})$ & - \\
\cline{2-4}
 & $T_{r1} < 0.5$, $T_{r2} < 0.5$ & $[x_{b1}, x_{b2})$ & General Rejection \\
\hline
  & $T_{r1} \geq 1-\min(p(t_1|x))$ & $\emptyset$ & ``{\it Majority-taking-all}'' Classification \\
  \cline{2-4}
  & $T_{r1} < 1-\min(p(t_1|x))$ & & ``{\it Majority-class and Reject-class}''\\
 Zero & $T_{r2} < 1-\max(p(t_2|x))$ & $[x_{b1}, x_{b4}]$ & Classification\\
 \cline{2-4}
 (Fig. 1d) & $T_{r1} < 1-\min(p(t_1|x))$ & & \\
  & $T_{r2} > 1-\max(p(t_2|x))$ & $[x_{b1}, x_{b2})$ and $(x_{b3}, x_{b4}]$ & General Rejection\\
 \cline{2-4}
 & $T_{r1} = 0$ & & ``{\it Minority-class and Reject-class}'' \\
 & $T_{r2} > 1-\max(p(t_2|x))>0.5$ & $(-\infty, x_{b2})$ and $(x_{b3}, \infty)$ & Classification\\
 \hline
 Zero, one and Two & & & \\
 (Fig.1) & $T_{r1} = T_{r2} = 0$ & $(-\infty, \infty)$ & Rejection to All\\
\hline
\end{tabular}
\end{table*}

Like the solution for boundary points, cross-over point(s) can also be obtained 
from solving eq. (\ref{equ:13}) or (19) by substituting $\delta_i=1$. One can prove that 
three specific cases will be met with the cross-over point(s) from the solution of 
eq. (\ref{equ:13}), namely, two, one, or zero cross-over point(s). The case for the two 
cross-over points appears only when $\alpha > 0$ in eq. (\ref{equ:14c}), and two curves 
of $p(t_1|x)$ and $p(t_2|x)$ demonstrate the non-monotonicity (Fig. 1b) through the equality 
$p(t_1|x)=1-p(t_2|x)$. When the associated standard deviations are equal for the two 
classes, i.e., $\sigma_1 = \sigma_2$, only one cross-over point appears, which corresponds 
to the monotonous curves of $p(t_1|x)$ and $p(t_2|x)$ (Fig. 1a). The case for the zero 
cross-over point occurs when $\alpha < 0$, which corresponds to no real-value (but complex-value) 
solution to eq. (\ref{equ:14a}) and to situations of non-monotonous curves of $p(t_1|x)$ 
and $p(t_2|x)$. In the followings, we will discuss several specific cases for rejections 
with respect to the cross-over points between the $p(t_1|x)$ and $p(t_2|x)$ curves, as well 
as to the associated settings on $T_r$ and $\lambda_{ij}$. A term is applied to 
describe every case. For example, ``Case$\_ k \_$BU'' indicates ``{\it k}'' for the $k$th 
case, ``B'' (or ``M'') for Bayesian (or mutual-information) classifiers, and ``G'' 
(or ``U'') for Gaussian (or uniform) distributions.
\newline
$Case\_1\_BG:$ {\it No rejection.}\newline
For a binary classification, Chow \cite{chow1970} showed that, when $T_{r1}=T_{r2}\geq0.5$, 
there exists no rejection for classifiers. The novel constraint of $T_{r1}+T_{r2}\leq1$ shown 
in eq. (\ref{equ:6e}) suggests that the setting should be $T_{r1}=T_{r2}=0.5$ when the thresholds 
are the input data. Users need to specify an option for ``{\it no rejection}'' 
or ``{\it rejection}'' as an input. When ``{\it no rejection}'' is selected, the 
conventional scheme of cost terms from a two-by-two matrix will be sufficient. 
Any usage of a two-by-three matrix will introduce some confusion that will be 
illustrated in the later section by Example 1. In addition, one cannot consider
$\lambda_{13} = \lambda_{23}=0$ as the defaults for the cost matrix in this case.
\newline
$Case\_2\_BG:$ {\it Rejection to all or to a complete class.}\newline
In discussing this case, we relax the constraints in eq. (\ref{equ:6e}) for 
including the zero values of the thresholds. Chow \cite{chow1970} showed that, 
whenever $T_r=0$, a classifier will reject all patterns. Substituting zero 
values for thresholds into eq. (\ref{equ:7}), one will obtain solutions 
for $\lambda_{11} = \lambda_{22} = \lambda_{13} = \lambda_{23} = 0$. These 
results imply that no cost is received even for a reject decision to a 
pattern. Obviously, a case like this should be avoided. 
In some situations, if one intends to reject a complete class 
(say, Class 1), its associated cost terms should be set to zero 
(say, $\lambda_{11} = \lambda_{13} = 0$). We call these situations 
as ``{\it one-class and reject-class}'' classification, since only 
two categories are identified, that is, ``Class 2'' and ``Reject Class'', 
respectively. 
\newline
$Case\_3\_BG:$ {\it Rejection in two cross-over points $x_{c1}$ 
and $x_{c2}$.}\newline
The necessary condition for realizing this case is derived from eq. 
(18) for $\alpha > 0$ while assuming $\delta_i=1$:
\setcounter{equation}{19}
\begin{equation}
 \label{equ:19}
\dfrac{\lambda_{12} - \lambda_{11}}{\lambda_{21} - \lambda_{22}} < 
\dfrac{p(t_2)\sigma_1}{p(t_1)\sigma_2}e^{\dfrac{\mu_1 - \mu_2}{2(\sigma^{2}_{1} - \sigma^{2}_{2})}}
\end{equation}
The general situation within this case is when $T_{r1}<0.5$ and $1-\max(p(t_2|x))<T_{r2}<0.5$, 
in which the reject region $R_3$ is divided by two ranges. 
When $T_{r1}<0.5$ and $T_{r2}<1-\max(p(t_2|x))<0.5$, only one 
class is identified, but all other patterns are classified into 
a reject class. Therefore, we refer this situation as ``{\it 
Class 1 and Reject-class}'' classification. Table I also lists 
the other situations for the rejections from the different settings on $T_{rj}$.
\newline
$Case\_4\_BG:$ {\it Rejection in one cross-over point $x_c$.}\newline
The general condition for realizing this case in the context of 
classifications is not based from setting an equality condition on (\ref{equ:19}) 
for $\alpha = 0$. We neglect such setting 
in this case, but assign it into Case$\_5\_$BG. As demonstrated in eq. 
(\ref{equ:14b}), the general condition of this case is a simply setting 
$\sigma_1 = \sigma_2$. Since the monotonicity property is enabled for
the curves of $p(t_1|x)$ and $p(t_2|x)$ in this case, a single reject region 
is formed (Fig. 1c).
\newline
$Case\_5\_BG:$ {\it Rejection in zero cross-over point.}\newline
The general condition for realizing this case corresponds to a violation 
of the criterion on (\ref{equ:14a}), or $\alpha < 0$ in (\ref{equ:19}). 
In this case, one class always shows a higher value of the posterior 
probability distribution over the other one in the whole domain of $x$. 
From definitions in the study of class imbalanced dataset \cite{japkowicz2002}
\cite{chawla2004}, if $p(t_1) > p(t_2)$ in binary classifications, Class 1 
will be called a ``{\it majority}'' class and Class 2 a ``{\it minority}'' class. 
Supposing that $p(t_1|x)>p(t_2|x)$, when $T_{r1}>1-\min(p(t_1|x))$, all patterns 
will be considered as Class 1. We call these situations 
as a ``{\it Majority-taking-all}'' classification. Due to the constraints 
like $T_{r1}+T_{r2} \leq 1$ and $p(t_1|x)+p(t_2|x)=1$, one is unable to 
realize a ``{\it Minority-taking-all}'' classification. 
When $T_{r1}<1-\min(p(t_1|x))$ and $T_{r2}<1-\max(p(t_2|x))$, all patterns 
will be partitioned into one of two classes, that is, majority and rejection. 
We call these situations ``{\it Majority-class and Reject-class}'' classifications. 
The situations of 
``{\it Minority-class and Reject-class}'' classification occur 
if $T_{r2}>1-\max(p(t_2|x))>0.5$ and $T_{r1}=0$.  

Since the study of imbalanced data learning received more attentions recently 
\cite{chawla2004}\cite{zhou2006}\cite{he2009}, 
one related theorem of Bayesian classifiers is derived below for elucidating their important features.

\begin{theorem}
 Consider a binary classification with an exact knowledge of one-dimensional 
Gaussian distributions. If a zero-one cost function is applied, Bayesian classifiers 
without rejection will satisfy the following rule:
\begin{equation}
 \label{equ:16}
 \begin{array}{r@{}}
if~ p_{min}=\min(p(t_1), p(t_2))\rightarrow 0, ~ and \\ 
~ \lambda_{11} = \lambda_{22} = 0, \lambda_{12} = \lambda_{21} = 1 \\
~ then ~E \rightarrow E_{max} = p_{min} ,  \end{array}
\end{equation}
which indicates that the classifiers have a tendency of reaching 
the maximum Bayesian error, $ E_{max}$,  by misclassifying all rare-class 
patterns in imbalanced data learning. 
\end{theorem}

\begin{proof}
We will prove the misclassification of all rare-class patterns first.
Suppose $p(t_2)$ represents the prior probability of the ``{\it minority}'' 
or ``{\it rare}'' class in imbalanced data learning and consider the special case 
firstly on the equal variances for two classes (Fig. 1a). When $p(t_2)$ approaches 
to zero, $x_c$ will approach infinity from using eq. (\ref{equ:14b}) with 
$\delta_i=1$. This result indicates that Bayesian classifiers will assign all 
patterns into the ``{\it majority}'' class in classifications. When the variances are 
not equal, eqs. (\ref{equ:14a}) and (\ref{equ:14c}) with $\delta_i=1$ will be applicable 
(Fig. 1b). One can obtain the relation $\alpha < 0$ for the case that no cross-over 
point occurs on $p(t_i|x)$ plots when $p(t_2)$ approaches to zero. Only the ``{\it majority}'' 
class is identified from using Bayesian classifiers in this case.
The equality of $ E_{max} = p_{min}$ suggests an upper bound of Bayesian error 
(See Appendix B). If violating this bound, Bayesian classifiers will adjust themselves
for achieving the smallest error rate. 
\end{proof}

\subsection{Examples of Bayesian Classifiers on Univariate 
Uniform Distributions}
\label{subsec:2.3}

\begin{figure*}[!t]
\centering
\includegraphics[width=4.5in,height=1.2in,bb= 60 0 450 100]{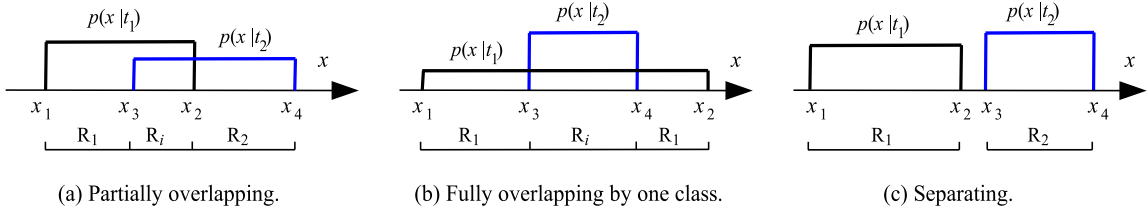}
\centering
\caption{Classification scenarios for univariate uniform distributions.}
\label{fig:2}
\end{figure*}

Chow \cite{chow1970} presented a study on rejection from Bayesian classifiers along 
uniform distributions for one-dimensional problems. This section will extend Chow's 
results by providing general formulas of parameterized distributions. A binary classification 
is considered. The two uniform distributions on two classes are given:
\begin{equation}\tag{22a}
 \label{equ:17a}
 p(x|t_1) = \left\{ \begin{array}{r@{\quad}l}
\dfrac{1}{x_2 - x_1} & when ~ x_1 \leq x \leq x_2\\
0~ \quad & otherwise
\end{array} \right.
\end{equation}

\begin{equation}\tag{22b}
 \label{equ:17b}
 p(x|t_2) = \left\{ \begin{array}{r@{\quad}l}
\dfrac{1}{x_4 - x_3} & when ~ x_3 \leq x \leq x_4\\
0~ \quad & otherwise
\end{array} \right.
\end{equation}
Three specific cases, shown in Fig. 2, will appear, namely, ``{\it Partially overlapping}'',  
``{\it Fully overlapping by one class}'', and ``{\it Separating}'' between two distributions 
for eq. (22). We will discuss each case with respect to their rejection settings. 
\newline
$Case\_1\_BU:$ {\it Partially overlapping between two distributions.}\newline
Suppose that the constraints for this case are:
\setcounter{equation}{22}
\begin{equation}
 \label{equ:18}
x_1 \leq x \leq x_4, ~and~ x_1 \leq x_3 \leq x_2 \leq x_4. 
\end{equation}
When the relevant knowledge of $p(t_i)$ and $p(x|t_i)$ is given, one is able 
to gain the posterior probabilities from eqs. (\ref{equ:1}) and (21) by a closed form:
\begin{equation}\tag{24a}
\label{equ:19a}
\begin{array}{r@{}l}
& p(t_1|x) = \\
& \left\{ \begin{array}{r@{~}l}
1 \qquad \qquad \qquad & when~ x_1 \leq x < x_3\\
\dfrac{p(t_1)(x_4 - x_3)}{p(t_1)(x_4 - x_3) + p(t_2)(x_2 - x_1)} & when~ x_3 \leq x \leq x_2\\
0 \qquad \qquad \qquad & otherwise
\end{array} \right.
\end{array}
\end{equation}

\begin{equation}\tag{24b}
\label{equ:19b}
\begin{array}{r@{}l}
& p(t_2|x) =\\
& \left\{ \begin{array}{r@{~}l}
1 \qquad \qquad \qquad & when~ x_2 \leq x < x_4\\
\dfrac{p(t_2)(x_2 - x_1)}{p(t_1)(x_4 - x_3) + p(t_2)(x_2 - x_1)} & when~ x_3 \leq x \leq x_2\\
0 \qquad \qquad \qquad & otherwise
\end{array} \right.
\end{array}
\end{equation}
Based on the Bayesian rules of eq. (\ref{equ:9}) and eq. (24), one can immediately 
determine $R_1 = [x_1, x_3)$ and $R_2 = [x_2, x_4]$ directly for Class 1 and Class 2, 
respectively, as shown in Fig. 2. The remaining range is denoted as $R_i = [x_3, x_2)$, 
since it needs to be identified further depending on the thresholds defined in (\ref{equ:7}). 
Due to the simplicity of the uniform distributions, one is able to realize analytical solutions 
directly for Bayesian classifiers. The probabilities of errors and rejects are calculated from :
\setcounter{equation}{24}
\begin{equation}
 \label{equ:20}
E = \left\{ \begin{array}{r@{~}l}
\dfrac{p(t_2)(x_2 - x_3)}{(x_4 - x_3)}, & ~ if ~ f(x \in R_i) = y_1\\
\dfrac{p(t_1)(x_2 - x_3)}{(x_2 - x_1)}, & ~ if ~ f(x \in R_i) = y_2\\
0,~ \qquad & ~ if ~ f(x \in R_i) = y_3
\end{array} \right.
\end{equation}
and 
\begin{equation}
 \label{equ:21}
Rej = \left\{ \begin{array}{r@{~}l}
& 0, \qquad  ~~~ if ~ f(x \in R_i) = y_1\\
& 0, \qquad ~~~ if ~ f(x \in R_i) = y_2\\
& (x_2 - x_3)[\dfrac{p(t_1)}{(x_2 - x_1)} + \dfrac{p(t_2)}{(x_4 - x_3)}],\\
& \qquad \qquad if ~ f(x \in R_i) = y_3
\end{array} \right.
\end{equation}
We use $f(x \in R_i) = y_j$ to describe a decision that $R_i$ is a range of Class $j$.   
Eq. (25) demonstrates that Bayesian classifiers with uniform distributions of 
classes will receive error either from Class 1 or from Class 2, but not both. When setting 
cost terms properly, zero error can be achieved with conditions of rejection on both classes 
as shown in eq. (26). It is interesting to observe that cost terms can only control 
the error types or give the appearance of rejection, but not the degree of them. This is significantly 
different from Bayesian classifiers with Gaussian distributions of classes. 
\newline
$Case\_2\_BU:$ {\it Fully overlapping by one class.}\newline
The constraints for this case are:
\begin{equation}
 \label{equ:22}
x_1 \leq x \leq x_2, ~and~ x_1 \leq x_3 \leq x_4 \leq x_2. 
\end{equation}
and the posterior probabilities are:
\begin{equation}\tag{28a}
\label{equ:23a}
\begin{array}{r@{}l}
& p(t_1|x) =\\
& \left\{ \begin{array}{r@{~}l}
1 \qquad \qquad \qquad & when~ x_1 \leq x < x_3\\
& or~ x_4 \leq x < x_2\\
\dfrac{p(t_1)(x_4 - x_3)}{p(t_1)(x_4 - x_3) + p(t_2)(x_2 - x_1)} & when~ x_3 \leq x \leq x_4\\
0 \qquad \qquad \qquad & otherwise
\end{array} \right.
\end{array}
\end{equation}

\begin{equation}\tag{28b}
\label{equ:23b}
\begin{array}{r@{}l}
& p(t_2|x) =\\
& \left\{ \begin{array}{r@{~}l}
\dfrac{p(t_2)(x_2 - x_1)}{p(t_1)(x_4 - x_3) + p(t_2)(x_2 - x_1)} & when~ x_3 \leq x \leq x_4\\
0 \qquad \qquad \qquad & otherwise
\end{array} \right.
\end{array}
\end{equation}
Following the similar way in the previous case, one can obtain the analytical results:
\setcounter{equation}{28}
\begin{equation}
 \label{equ:24}
E = \left\{ \begin{array}{r@{~}l}
p(t_2), \qquad & ~ if ~ f(x \in R_i) = y_1\\
\dfrac{p(t_1)(x_4 - x_3)}{(x_2 - x_1)}, & ~ if~ f(x \in R_i) = y_2\\
0,~ \qquad & ~ if~ f(x \in R_i) = y_3
\end{array}, \right.
\end{equation}
and
\begin{equation}
 \label{equ:25}
Rej = \left\{ \begin{array}{r@{~}l}
0, \qquad \qquad \qquad & ~ if~ f(x \in R_i) = y_1\\
0, \qquad \qquad \qquad & ~ if~ f(x \in R_i) = y_2\\
p(t_1)\dfrac{(x_4 - x_3)}{x_2 - x_1} + p(t_2), &~ if ~ f(x \in R_i) = y_3
\end{array}. \right.
\end{equation}
Specific solutions will be received in this case on Class 2, which is full overlapped 
within Class 1. All patterns within Class 2 may be misclassified or rejected fully 
depending on the settings of cost terms.
\newline
$Case\_3\_BU:$ {\it Separation between two distributions.}\newline
One is able to obtain the exact solutions without any error and reject. Cost terms are useless in this case.

\section{Mutual-information based Classifiers with A Reject Option}
\label{sec:3}
\subsection{Mutual-information based Classifiers}
\label{subsec:3.1}

\begin{definition}[Mutual-information classifier]
 A mutual-information classifier is the classifier which is obtained from
the maximization of mutual information over all patterns: 
\begin{equation}
 \label{equ:31}
y^+ = arg \max \limits_{y}NI(T = t, Y = y),
\end{equation}
where $T$ and $Y$ are the target variable and decision 
output variable, $t$ and $y$ are their values, respectively. For simplicity, 
we denote $NI(T=t, Y=y) = NI(T, Y)$ as  the normalized mutual information in a form of \cite{hu2008}:
\begin{equation}\tag{32a}
\label{equ:32a}
 NI(T, Y) = \dfrac{I(T, Y)}{H(T)}
\end{equation}
where $H(T)$ is the entropy based on the Shannon definition \cite{shannon1948} to the target variable, 
\begin{equation}\tag{32b}
\label{equ:32b}
 H(T) = - \sum \limits^{m}_{i = 1}p(t_i)log_{2}p(t_i),
\end{equation}
and $I(T,Y)$ is mutual information between two variables of $T$ and $Y$ \cite{cover2006}:
\begin{equation}\tag{32c}
\label{equ:32c}
I(T, Y) = \sum \limits_{i=1}^{m} \sum \limits_{j=1}^{m+1} p(t_i, y_j)log_{2}\dfrac{p(t_i, y_j)}{p(t_i)p(y_j)},
\end{equation}
where $m$ is a total number of classes in $T$. For binary classifications, we set $m=2$. 
In (32), $p(t, y)$ is the joint distribution between the two variables, and $p(t)$ and $p(y)$ are 
the marginal distributions which can be derived from \cite{cover2006}:
\setcounter{equation}{32}
\begin{equation}
 \label{equ:33}
p(t) = \sum \limits_{y} p(t, y), ~and~~p(y) = \sum \limits_{t}p(t, y).
\end{equation}
\end{definition}

Mathematically, eq. (\ref{equ:31}) expresses that $y^+$ is an optimal classifier in terms of 
the maximal mutual information, or relative entropy, between the target variable $T$ and 
decision output variable $Y$. The physical interpretation of relative entropy is a measurer 
of probability similarity between the two variables. Note that the present definition of 
$NI$ is asymmetry to the variables $T$ and $Y$ for the normalization term of 
$H(T)$ (=constant, for given $p(t))$, but will not make a difference for arriving at the 
optimal $y^+$ defined by (\ref{equ:31}). We adopt Shannon's definition of entropy for the 
reason that no free parameter is introduced. A normalization scheme is applied so that 
a relative comparison can be made easily among classifiers.

\begin{definition}[Augmented Confusion Matrix \cite{hu2008}]
 An augmented confusion matrix will include one column for a rejected class, which is 
added on a conventional confusion matrix: 
\begin{equation}
 \label{equ:34}
C = \left[ \begin{array}{r@{\cdots}l}
c_{11} \quad c_{12} \quad & \quad c_{1m} \quad c_{1(m+1)}\\
c_{21} \quad c_{22} \quad & \quad c_{2m} \quad c_{2(m+1)}\\
& \\
c_{m1} \quad c_{m2} \quad & \quad c_{mm} \quad c_{m(m+1)}
\end{array}\right],
\end{equation}
where $c_{ij}$ represents the number of the $i$th class that is classified as the $j$th class. 
The row data corresponds to the exact classes, and the column data corresponds to the 
prediction classes. The last column represents a reject class. The relations and constraints of an augmented 
confusion matrix are:
\begin{equation}
 \label{equ:35}
C_i = \sum \limits_{j=1}^{m+1}c_{ij}, ~C_i>0, ~c_{ij} \geq 0, ~i = 1,2, \cdots,m
\end{equation}
where $C_i$ is the total number for the $i$th class. The data for $C_i$ is known in classification problems.
\end{definition}

In this work, supposing that the input data for classifications are exactly known about 
the prior probability $p(t_i)$ and the conditional probability density function $p(\textbf{x}|t_i)$, 
one is able to derive the joint distribution matrix in association with the confusion matrix:
\begin{equation}
 \label{equ:36}
\begin{array}{r@{}l}
& p_{ij} = p(t_i, y_j) = \int \limits_{R_j}p(t_i)p(\textbf{x}|t_i)d\textbf{x} \approx \dfrac{c_{ij}}{n} = p_e(t_i, y_j),\\
& i = 1,2, \cdots, m, ~j = 1,2, \cdots, m+1
\end{array}
\end{equation}
where $R_j$ is denoted as the region in which every pattern $\textbf{x}$ is identified as the $j$th class, 
and  $p_e(t_i, y_j)$ is the empirical probability density for applications where only a 
confusion matrix is given. In those applications, the total number of patterns $n$ is generally known. 

Eq. (\ref{equ:36}) describes the approximation relations between the joint distribution 
and confusion matrix. If the knowledge about $p(t_i)$ and $p(\textbf{x}|t_i)$ are exactly known, 
one can design a mutual information classifier directly. If no initial information is known about 
$p(t_i)$ and $p(\textbf{x}|t_i)$, the empirical probability density of joint distribution, 
$p_e(t_i, y_j)$, can be estimated from the confusion matrix \cite{hu2008}. 
This treatment, based on the frequency principle of a confusion matrix, is not 
mathematically rigorous, but will offer a simple approach for classifiers to apply 
the entropy principle for wider applications.

\begin{figure*}[!t]
\begin{equation}\tag{42a}
\label{equ:42a}
\begin{array}{r@{}l}
p(t,y) = 
 \left[ \begin{array}{r@{}l}
\dfrac{p(t_1)}{2}[1-erf(X_{11})] ~~~ \dfrac{p(t_1)}{2}[1-erf(X_{12})] ~~~ \dfrac{p(t_1)}{2}[erf(X_{11})+erf(X_{12})]\\
\dfrac{p(t_2)}{2}[1-erf(X_{21})] ~~~ \dfrac{p(t_2)}{2}[1-erf(X_{22})] ~~~ \dfrac{p(t_2)}{2}[erf(X_{21})+erf(X_{22})]
\end{array}\right],
\end{array}
\end{equation}
\end{figure*}

Considering binary classifications, one will have the following formula for the joint 
distribution $p(t, y)$: 
\begin{equation}
 \label{equ:37}
\begin{array}{r@{}l}
& p(t,y) = \\
& \left[ \begin{array}{r@{}l}
& \int \limits_{R_1}p(t_1)p(\textbf{x}|t_1)d\textbf{x} ~ \int \limits_{R_2}p(t_1)p(\textbf{x}|t_1)d\textbf{x} ~ \int \limits_{R_3}p(t_1)p(\textbf{x}|t_1)d\textbf{x}\\
& \int \limits_{R_1}p(t_2)p(\textbf{x}|t_2)d\textbf{x} ~ \int \limits_{R_2}p(t_2)p(\textbf{x}|t_2)d\textbf{x} ~ \int \limits_{R_3}p(t_2)p(\textbf{x}|t_2)d\textbf{x}
\end{array}\right].
\end{array}
\end{equation}
The marginal distribution for $p(t)$ is in fact the given information of prior knowledge about the classes:
\begin{equation}
 \label{equ:38}
p(t) = (p(t_1), p(t_2))^T
\end{equation}
where the superscript ``$T$'' represents a transpose, and the marginal distribution for $p(y)$ is: 
\begin{equation}
 \label{equ:39}
\begin{array}{r@{}l}
& p(y) = (p(y_1), p(y_2), p(y_3)) = (\int \limits_{R_1}Qd\textbf{x}, \int \limits_{R_2}Qd\textbf{x}, \int \limits_{R_3}Qd\textbf{x})\\
& Q = p(t_1)p(\textbf{x}|t_1) + p(t_2)p(\textbf{x}|t_2).
\end{array}
\end{equation}
Substituting (\ref{equ:37}) - (\ref{equ:38}) into (32), one can obtain the formula 
of $NI$ in terms of $p(t_i)$ and $p(\textbf{x}|t_i)$. When the prior knowledge of $p(t_i)$ 
is given, the conditional entropy $H(T)$ in eq. (\ref{equ:32b}) will be unchanged during 
classifier learnings. This is why we use this term to normalize the mutual information 
in (\ref{equ:32a}).  

\subsection{Examples of Mutual-information Classifiers on Univariate 
Gaussian Distributions}
\label{subsec:3.2}
Mutual-information classifiers, like Bayesian classifiers, also provide a general 
formulation to classifications. They are able to process classifications
with or without rejection. 
This section will aim at deriving novel formulas necessary 
for the design and analysis of mutual-information classifiers under assumptions of Gaussian 
distributions. The assumptions, or given input data, for the derivations are kept the same 
as those for Bayesian classifiers shown in Section II, except that cost terms 
of $\lambda_{ij}$ are not given as the input, but will be displayed as the output of 
the classifiers. In other words, mutual information classifiers will automatically calculate 
the two thresholds that can lead to the cost terms through eq. (\ref{equ:7}). However, due to 
a redundancy among six cost terms, one will fail to obtain the unique solution of the cost 
terms, which is demonstrated in Example 1 of Section IV.

\begin{table*}
 \label{tab:2}
\renewcommand{\arraystretch}{1.3}
\caption{Classification Regions for Mutual-information Classifiers in Univariate Gaussian Distributions 
of Fig. 1 $(x_{b1}<x_{b2}<x_{b3}<x_{b4})$}
\setlength{\tabcolsep}{1pt}
\centering
\begin{tabular}{|c|c|c|c|c|c|}
 \hline
Reject Option & Cross-over Point(s) & Boundary Point(s) & Class of $R_1$ 
& ~~~~~~~~Class of $R_2 ~~~~~~~~$ & Class of $R_3$\\
\hline
 \multirow{2}{1in}{\centering No Rejection} & $x_c$ & $x_b$ & $(-\infty, x_b)$ 
& $[x_b, \infty)$ & $\emptyset$\\ 
						\cline{2-6}
								      & $x_{c1}, x_{c2}$ 
& $x_{b1}, x_{b2}$ & $(-\infty, x_{b1})$ and $(x_{b2}, \infty)$ & $[x_{b1}, x_{b2}]$ & $\emptyset$\\ 
\hline
\multirow{2}{1in}{\centering Rejection} & $x_c$ & $x_{b1}, x_{b2}$ & $(-\infty, x_{b1})$ 
& $[x_{b2}, \infty)$ & $[x_{b1}, x_{b2})$\\ 
						\cline{2-6}
								      & $x_{c1}, x_{c2}$ 
& $x_{b1}, x_{b2}, x_{b3}, x_{b4}$ & $(-\infty, x_{b1})$ and $(x_{b2}, \infty)$ 
& $[x_{b2}, x_{b3}]$ & $[x_{b1}, x_{b2})$ and $(x_{b3}, x_{b4}]$\\ 
\hline
\end{tabular}
\end{table*}

Generally, one is unable to derive a closed-form solution to mutual-information classifiers. 
One of the obstacles is the nonlinear complexity of solving error functions. Therefore, 
this work only provides semi-analytical solutions for mutual information classifiers. 
When substituting $p(t_i)$ and $p(x|t_i)$ into eqs. (\ref{equ:31}) and (32), one will 
encounter the process of solving an inverse problem on the following function:
\begin{equation}
 \label{equ:40}
\max \limits_{y \in Y} NI(T, Y) = \max f(x, \mathbf{\theta} = (p(t_i), p(x|t_i), x_{bj})),
\end{equation}
for searching the boundary points $x_{bj}$ from error functions. Only numerical solutions 
can be obtained for $x_{bj}$, except for a special case. Whenever a reject option is set, 
mutual-information classifiers will generate classification regions, $R_i~ (i=1,2,3)$, 
automatically according to the given data of $p(t_i)$ and $p(x|t_i)$, as shown in Table II. 
In the followings, some specific cases of mutual-information classifiers will be discussed 
in related to a reject option.
\newline
$Case\_1\_MG:$ {\it No rejection in one cross-over point $x_c$ when $p(t_1)=p(t_2)$ 
and $\sigma_1=\sigma_2$.}\newline
This is a very special case where one is able to obtain a closed-form solution to 
mutual-information classifiers. Under the conditions of $p(t_1)=p(t_2)$, $\sigma_1=\sigma_2$, 
and two by two joint distribution matrix for no rejection, one can get a single boundary 
point $x_b$, coincident to the cross-over point $x_c$, for partitioning the classification regions:
\begin{equation}
 \label{equ:41}
\begin{array}{r@{}l}
& \qquad \qquad x_b = x_c = \dfrac{\mu_1 + \mu_2}{2},\\
& if~ \mu_1 < \mu_2 ~then~ R_1 = (-\infty, x_b), R_2 = [x_b, \infty), R_3 = \emptyset.
\end{array}
\end{equation}
This result exhibits similar results for Bayesian classifiers, which leads to the same error 
values between the two types of classifiers. Therefore, eq. (\ref{equ:41}) indicates that
$y^+ = y^*$ to be fully equivalent between mutual-information classifiers and Bayesian 
classifiers under the conditions of $p(t_1)=p(t_2)$ and $\sigma_1=\sigma_2$ when no reject 
option is selected. 
\newline
$Case\_2\_MG:$ {\it Rejection in one cross-over point $x_c$ and $\sigma_1=\sigma_2$.}\newline
When we relax the condition in the case above on $p(t_1) \neq (t_2)$ and with a reject 
option, the solutions to mutual-information classifiers become not fully analytical. 
The key step for missing such an analytical solution comes from a determination of 
$x_{bj}$. In this case, due to the condition that $\sigma_1=\sigma_2$, one will have 
a single cross-over point $x_c$ as the general case in binary classifications for 
Gaussian distributions. If a reject option is selected, one will generally have two 
boundary points $x_{b1}$ and $x_{b2}$. Suppose $\mu_1<\mu_2$ and $x_{b1}<x_{b2}$, one 
can partition classification regions as: $R_1 = (-\infty, x_{b1}), R_2 = [x_{b2}, 
\infty)$, and $R_3 = [x_{b1}, x_{b2})$. Supposing the two boundary points are given, 
one can have a closed-form formula on eq. (\ref{equ:37}):

                 (Please see the equation on the top of this page)
\\
where $erf(\cdot)$ is an error function, and 
\begin{equation}\tag{42b}
 \label{equ:42b}
X_{ij} = \dfrac{\mu_i - x_{bj}}{\sqrt{2}\sigma_i},~i = 1,2, ~j = 1,2.
\end{equation}
In this work, we adopt a numerical approach to search the results on $x_{b1}$ and $x_{b2}$. 
Whenever these values are known, one can get the error rate and reject 
rate from:
\begin{equation}\tag{43a}
 \label{equ:43a}
\begin{array}{r@{}l}
& E = E_1 + E_2 = p(t_i =1, y_j = 2) + p(t_i = 2, y_j = 1)\\
& \quad = \dfrac{p(t_1)}{2}[1-erf(X_{12})] + \dfrac{p(t_2)}{2}[1-erf(X_{21})]
\end{array}
\end{equation}

\begin{equation}\tag{43b}
 \label{equ:43b}
\begin{array}{r@{}l}
& Rej = Rej_1 + Rej_2 \\
& \quad = p(t_i =1, y_j = 3) + p(t_i = 2, y_j = 3)\\
& \quad = \dfrac{p(t_1)}{2}[erf(X_{11})+erf(X_{12})] \\
& \qquad \qquad \qquad + \dfrac{p(t_2)}{2}[erf(X_{21})+erf(X_{22})]
\end{array}
\end{equation}
\newline
$Case\_3\_MG:$ {\it Rejection in two cross-over points.}\newline
This is a general case for mutual-information classifiers 
in which four boundary points, $x_{bj}$, are formed. When the four points 
obtained numerically during solving eq. (\ref{equ:31}), the classification regions 
$R_1$ to $R_3$ will be set as shown in Table II. With the condition of 
$x_{b1}<x_{b2}<x_{b3}<x_{b4}$, the closed-form solution of $p(t, y)$ can be given 
in a similar way of eq. (42). Additionally, both error and reject rates can be evaluated 
from $p(t, y)$. For comparing with Bayesian classifiers, the equivalent rejection 
thresholds are derived from the given data of $x_{bj}$:
\begin{equation}\tag{44a}
 \label{equ:44a}
\begin{array}{r@{}l}
& T_{r1} = 1 - p(t_1|x = x_{b1}) \\
&= 1 - \dfrac{p(t_1)\sigma_2 e^{\dfrac{-(x_{b1} - \mu_1)^2}{2\sigma^{2}_{1}}}}{p(t_1)
\sigma_2 e^{\dfrac{-(x_{b1} - \mu_1)^2}{2\sigma^{2}_{1}}} + p(t_2)\sigma_1 e^{\dfrac{-(x_{b1} 
- \mu_2)^2}{2\sigma^{2}_{2}}}}
\end{array}
\end{equation}

\begin{equation}\tag{44b}
 \label{equ:44b}
\begin{array}{r@{}l}
& T_{r2} = 1 - p(t_2|x = x_{b2}) \\
&= 1 - \dfrac{p(t_2)\sigma_1 e^{\dfrac{-(x_{b2} - \mu_2)^2}{2\sigma^{2}_{2}}}}{p(t_1)
\sigma_2 e^{\dfrac{-(x_{b2} - \mu_1)^2}{2\sigma^{2}_{1}}} + p(t_2)\sigma_1 e^{\dfrac{-(x_{b2} 
- \mu_2)^2}{2\sigma^{2}_{2}}}}
\end{array}
\end{equation}
With the condition of $x_{b1}<x_{b2}<x_{b3}<x_{b4}$ shown in Fig. 1d, substituting either $x_{b1}$ 
or $x_{b4}$ into (44) will give the same value on $T_{r1}$, and a similar one for $x_{b2}$ 
or $x_{b3}$ on $T_{r2}$. The results of $T_{r1}$ and $T_{r2}$ indicate that mutual-information 
classifiers will automatically search the rejection thresholds for balancing the error rate and 
reject rate for the given data of classes. This specific feature will be discussed in Section IV.

\subsection{Examples of Mutual-information Classifiers on Univariate 
Uniform Distributions}
\label{subsec:3.3}
When comparing with Bayesian classifiers, we examine mutual-information classifiers on uniform 
distributions in this section. The two classes and their conditional probability density functions 
are given in (22). Three cases will be discussed below.
\newline
$Case\_1\_MU:$ {\it Partially overlapping between two distributions.}\newline
In this case (Fig. 2a), one needs to construct joint distribution $p(t, y)$ first. For binary 
classifiers, $p(t, y)$ is given in the following forms:
\begin{equation}\tag{45a}
 \label{equ:45a}
\begin{array}{r@{}l}
& p(t,y)  =  \left[ \begin{array}{r@{~}l}
p(t_1) \qquad \qquad 0 \qquad \qquad & 0\\
\dfrac{p(t_2)(x_2 - x_3)}{(x_4 - x_3)} ~ \dfrac{p(t_2)(x_4 - x_2)}{(x_4 - x_3)} & 0
\end{array}  \right], \\
&\qquad \qquad \qquad \qquad if ~ {f(x \in R_i)  = y_1} 
\end{array}
\end{equation}

\begin{equation}\tag{45b}
 \label{equ:45b}
\begin{array}{r@{}l}
& p(t,y)= \left[ \begin{array}{r@{~}l}
\dfrac{p(t_1)(x_3 - x_1)}{(x_2 - x_1)} ~ \dfrac{p(t_1)(x_2 - x_3)}{(x_2 - x_1)} & 0\\
0 \qquad \qquad p(t_2) \qquad \qquad & 0
\end{array}\right], \\
&\qquad \qquad \qquad \qquad if ~ {f(x \in R_i)  = y_2} 
\end{array}
\end{equation}

\begin{equation}\tag{45c}
 \label{equ:45c}
\begin{array}{r@{}l}
& p(t,y) = \left[ \begin{array}{r@{~}l}
\dfrac{p(t_1)(x_3 - x_1)}{(x_2 - x_1)} ~ \qquad 0 \qquad \qquad & \dfrac{p(t_1)(x_2 - x_3)}{(x_2 - x_1)}\\
\qquad \qquad 0 \qquad ~ \dfrac{p(t_2)(x_4 - x_2)}{(x_4 - x_3)} & \dfrac{p(t_2)(x_2 - x_3)}{(x_4 - x_3)}
\end{array}\right], \\
&\qquad \qquad \qquad \qquad if ~ {f(x \in R_i)  = y_3} 
\end{array}
\end{equation}
Eq. (45) demonstrates three sets of $p(t, y)$ due to diffident decisions may be involved 
with $R_i$ in Fig. 2a. Substituting (45) into (32), one will obtain three sets of $NI$'s. 
The closed-form solutions about the decision can be given, but this 
work adopts a numerical approach for omitting tedious descriptions of the formulas. 
\newline
$Case\_2\_MU:$ {\it Fully overlapping by one class.}\newline
The formula for  $p(t, y)$ in this case (Fig. 2b) is:
\begin{equation}\tag{46a} 
 \label{equ:46a}
p(t,y) = \left[ \begin{array}{r@{~}l}
p(t_1) \quad 0 \quad & 0\\
p(t_2) \quad 0 \quad & 0
\end{array}\right], 
~ if ~ f(x \in R_i) = y_1
\end{equation}
\begin{equation}\tag{46b}
 \label{equ:46b}
\begin{array}{r@{}l}
& p(t,y)= \\
& \left[ \begin{array}{r@{~}l}
 \dfrac{p(t_1)(x_2 - x_1 - x_4 + x_3)}{(x_2 - x_1)} ~ \dfrac{p(t_1)(x_4 - x_3)}{(x_2 - x_1)} & 0\\
 0 \qquad \qquad p(t_2) \qquad \qquad & 0
\end{array}\right], \\
& \qquad  \qquad  \qquad ~ if ~ f(x \in R_i) = y_2 
\end{array}
\end{equation}
\begin{equation}\tag{46c}
 \label{equ:46c}
\begin{array}{r@{}l}
& p(t,y) =  \\
& \left[ \begin{array}{r@{~}l}
\dfrac{p(t_1)(x_2 - x_1 - x_4 + x_3)}{(x_2 - x_1)} ~~ 0 ~  & \dfrac{p(t_1)(x_4 - x_3)}{(x_2 - x_1)}\\
 \qquad \qquad 0 \qquad \qquad ~ \quad ~ 0~  &  \qquad  p(t_2)
\end{array}\right], \\
& \qquad  \qquad  \qquad ~ if ~ f(x \in R_i) = y_3
\end{array}
\end{equation}
One can get the following results through substituting (46) into (32):
\begin{equation}\tag{47a}
\label{equ:47a}
NI(t, y) = 0, ~if ~f(x \in R_i)  = y_1
\end{equation}
\begin{equation}\tag{47b}
\label{equ:47b}
0 < NI(t, y) < 1, ~if ~f(x \in R_i)  = y_2 (or~ y_3).
\end{equation}

Eq. (47a) suggests that the decision for $f(x \in R_i)  = y_1$ will produce zero information.
Therefore, mutual information classifiers will never make this kind of decisions (but Bayesian
classifiers may do so).
\newline
$Case\_3\_MU:$ {\it Separation between two distributions.}\newline
Mutual-information classifiers will show the perfect solutions as those for Bayesian classifiers. 

\section{Comparisons between Bayesian Classifiers and Mutual-information Classifiers}
\label{sec:4}
\subsection{General Comparisons}
\label{subsec:4.1}

\begin{table*}
 \label{tab:3}
\renewcommand{\arraystretch}{1.3}
\caption{Data Information for Bayesian and Mutual-information Classifiers in Binary Classifications}
\setlength{\tabcolsep}{5pt}
\centering
\begin{tabular}{|c|c|c|c|c|}
 \hline
Classifier & \multicolumn{2}{p{3.5cm}|}{\centering Required Input} & Learning & Output\\
\cline{2-3}
\multicolumn{1}{|c|}{Type} & \multicolumn{1}{c|}{On Data} & \multicolumn{1}{c|}{On Rejection} 
& \multicolumn{1}{c|}{Target} & \multicolumn{1}{c|}{Data}\\
\hline
 \multirow{6}{1in}{\centering Bayesian} & $p(t_1)$, $p(t_2)$ & & & $E_1$, $E_2$, $Rej_1$, $Rej_2$,\\ 
								  & $p(\textbf{x}|t_1)$, $p(\textbf{x}|t_2)$ & No & $\min Risk(y)$ & $Risk$,\\
								  & $\lambda_{11}$, $\lambda_{12}$, $\lambda_{13}$ & or & or & $R_1$, $R_2$, $R_3$,\\
								  & $\lambda_{21}$, $\lambda_{22}$, $\lambda_{23}$ & Yes & $\min E(y)$ & $T_{r1}$, $T_{r2}$,\\
								  & (or $T_{r1}$, and $T_{r2}$) &  &  & $(\{\lambda_{21}/\lambda_{12}\}$, or\\
								  & &  & & $\{\lambda_{21}, \lambda_{13}, \lambda_{23}\})$\\
\hline
\multirow{6}{1in}{\centering Mutual-\\Information} &  & &  & $E_1$, $E_2$, $Rej_1$, $Rej_2$,\\ 
										&  & No & & $NI$,\\ 
				  & $p(t_1)$, $p(t_2)$ & or & $\max NI(T, Y)$ & $R_1$, $R_2$, $R_3$,\\ 
				  & $p(\textbf{x}|t_1)$, $p(\textbf{x}|t_2)$ & Yes & & $T_{r1}$, $T_{r2}$, \\ 
				  &  & &  & $(\{\lambda_{21}/\lambda_{12}\}$, or\\ 
				  &  & &  & $\{\lambda_{21}, \lambda_{13}, \lambda_{23}\})$\\ 
\hline
\end{tabular}
\end{table*}

Mutual-information classifiers provide users a wider perspective in processing classification 
problems, hence a larger toolbox in their applications. For discovering new features in this 
approach, the present section will discuss general aspects of mutual-information and Bayesian 
classifiers at the same time for a systematic comparison. The main objective of the comparative 
study is to reveal their corresponding 
advantages and disadvantages. Meanwhile, their associated issues, or new challenges, are also 
presented from the personal viewpoint of the author.

First, both types of classifiers share the same assumptions of requiring the exact knowledge 
about class distributions and specifying the status of the reject option (Table III). 
The ``{\it exact knowledge}'' feature imposes the most weakness on the two approaches 
in applications. In other words, the approaches are more theoretically meaningful, rather 
than directly useful for solving real-world problems. When the exact knowledge is not available, 
the existing estimation approaches to class distributions 
\cite{duda2001}\cite{fukunaga1990}\cite{berger1985} for Bayesian classifiers will be feasible 
for implementing mutual-information classifiers. The learning targets of Bayesian classifiers 
involve evaluations of risks or errors, which is mostly compatible with classification 
goals in real-life applications. However, the concept of mutual information, or 
entropy-based criteria, is not a common concern or requirement from most classifier 
designers and users \cite{hu2008}.

Second, Bayesian classifiers will ask (or implicitly apply) cost terms for their designs. 
This requirement provides both advantages and disadvantages depending on applications.
The main advantage is its flexibility in offering objective or subjective 
designs of classifiers. When the exact knowledge is available and reliable, inputing such 
data will be very simple and meaningful for realizing objective designs. At the same time, 
subjective designs will always be possible. The main disadvantage may occur for objective 
designs if one has incomplete information about cost terms. Generally, cost terms are more 
liable to subjectivity than prior probabilities. In this case, avoiding the introduction of 
subjectivity is not an easy task for Bayesian classifiers. Mutual-information classifiers, 
without requiring cost terms, will fall into an objective approach. They carry an intrinsic 
feature of ``{\it letting the data speak for itself}'', which exhibits a significant 
difference from a subjective version of Bayesian classifiers. However, the current 
definition of mutual-information classifiers needs to be extended for carrying 
the flexibility of subjective designs, which is technically feasible by introducing 
free parameters, such as fuzzy entropy \cite{liu2007}. 

Third, one of the problems for the current learning targets of Bayesian classifiers is 
their failure to obtain the optimal rejection threshold in classifications. Although 
Chow \cite{chow1970} and Ha \cite{ha1997} suggested formulas respectively in forms of:
\setcounter{equation}{47}
\begin{equation}
 \label{equ:49a}
\min Risk(T_r) = E(T_r) + T_r Rej(T_r),
\end{equation}
or
\begin{equation}
 \label{equ:49b}
\min \dfrac{E(T_r)}{Rej(T_r)},
\end{equation}
respectively, a minimization from both formulas will lead to a solution of $T_r=0$ 
for $Risk=0$, which implies a rejection of all patterns. Therefore, we can expect to 
establish a meaningful learning target which is applicable to Bayesian classifiers 
for determining optimal rejection thresholds. On the contrary, mutual-information 
classifiers are able to achieve the optimal rejection thresholds as the classifiers' outcomes. 
The remaining issue is to study the optimal cases in a systematic way.

Fourth, Bayesian classifiers generally fail to handle class imbalanced data properly if no 
cost terms are specified in classifications, as described in Theorem 3. When one 
class approximates a smaller (or zero) population and no distinction is made among 
error types, Bayesian classifiers have a tendency to put all patterns of the smaller 
class into error, and its NI will be approximately zero, which represents that no information 
is obtained from classifiers \cite{mackay2003}. Mutual-information classifiers display 
particular advantages in these situations, including cases for abstaining classifications. 
They provide a solution of balancing error types and reject types without using cost terms. 
The challenge lies in their theoretical derivation of response behaviors, such as, upper bound 
and lower bound of $E_i/p(t_i)$ for mutual-information classifiers. 

Fifth, mutual-information classifiers will add extra computational complexities and costs 
over Bayesian classifiers. Both types of classifiers require computations of posterior 
probability. When these data are obtained, Bayesian classifiers will produce decision 
results directly. However, mutual-information classifiers will need further procedures,
such as, to form a confusion matrix (or a joint distribution matrix), to evaluate $NI$ 
in (\ref{equ:31}), and to search boundary points from a non-convex space $NI$ in 
(\ref{equ:40}). These procedures will introduce significantly analytical and computational 
difficulties to mutual-information classifiers, particularly in multiple-class problems 
with high dimensions. 

Note that the discussions above provide a preliminary answer to the question posed in 
the title of this paper. In another connection, Appendix B presents the tighter 
bounds between conditional entropy and Bayesian error in binary 
classifications. Further investigations are expected to search other differences 
under various assumptions or backgrounds, such as distributions of mixture models, 
multiple-class classifications in high dimension variables, rejection to a subset of 
classes \cite{ha1997}, and experimental studies from real-world datasets. 

\subsection{Comparisons on Univariate Gaussian Distributions}
\label{subsec:4.2}
Gaussian distributions are important not only in theoretical sense. To a large extent, 
this assumption is also appropriate for providing critical guidelines in real applications. 
For classification problems, many important findings can be revealed from a study on 
Gaussian distributions. 

The following numerical examples are specifically designed for 
demonstrating the intrinsic differences between Bayesian and mutual-information classifiers 
on Gaussian distributions. For calculations of $NI$'s values on the following example, 
an open-source toolkit \cite{hu2009b} is adopted for computations of mutual-information classifiers. 

\begin{example}
{\it Two cross-over points}. The data for no rejection are given below:  
\begin{equation}
 \begin{array}{r@{}l}
& No~rejection: \\
& \mu_1=-1, ~\sigma_1=2,~p(t_1)=0.5,~\lambda_{11}=0,~\lambda_{12}=1,\\ \nonumber
& \mu_2=1, ~~ ~\sigma_2=1,~p(t_2)=0.5,~\lambda_{21}=1,~\lambda_{22}=0
\end{array}
\end{equation}
The cost terms are used for Bayesian classifiers, but not for mutual-information 
classifiers. Table IV lists the results for both classifiers. One can obtain the same 
results when inputing $\lambda_{13}=1-\lambda_{23}$ for Bayesian classifiers. This is 
why a two-by-two matrix has to be used in the case of no rejection. Two cross-over points are 
formed in this examples (Fig. 1b). If no rejection is selected, both classifiers will 
have two boundary points. Bayesian classifiers will partition the classification regions 
by having $x_{b1}=x_{c1}=-0.238$ and $x_{b2}=x_{c2}=3.57$. Mutual-information classifiers 
widen the region $R_2$ by $x_{b1}=-0.674$ and $x_{b2}=4.007$ so that the error for 
Class 2 is much reduced. If considering zero costs for correct 
classifications and using eq. (\ref{equ:13}) with  $\delta_i=\lambda_{21}/\lambda_{12}$, 
one can calculate a cost ratio below for an independent parameter to Bayesian classifiers 
in the case of no rejection:
\begin{equation}
 \label{equ:49}
\Lambda_{21} = \dfrac{\lambda_{21}}{\lambda_{12}} = \dfrac{p(x=x_b|t_1)p(t_1)}{p(x=x_b|t_2)p(t_2)},
\end{equation}
which is used to establish an equivalence between mutual-information classifiers and 
Bayesian classifiers. Substituting the boundary points of mutual-information classifier 
at $x_{b1}=-0.674$ and $x_{b2}=4.007$ into $p(x|t_i)$ and (50), respectively, 
one receives a unique cost ratio value, $\Lambda_{21}=2.002$. Hence, this mutual-information 
classifier has its unique equivalence to a specific Bayesian classifier which is exerted by 
the following conditions to the cost terms:
\begin{equation}
 \lambda_{11} = 0,~\lambda_{12}=1.0,~\lambda_{21}=2.002,~\lambda_{22}=0. \nonumber
\end{equation}
\end{example}

\begin{table*}
 \label{tab:4}
\renewcommand{\arraystretch}{1.3}
\caption{Results of Example 1 on Univariate Gaussian Distributions}
\setlength{\tabcolsep}{7pt}
\centering
\begin{tabular}{|c|c|c|c|c|c|c|c|c|}
 \hline
Reject & Classifier & $E_1$ & & $Rej_1$ & & $T_{r1}$ & $x_{b1}$, $x_{b2}$ & \\
Option & Type & $E_2$ & $E$ & $Rej_2$ & $Rej$ & $T_{r2}$ & $x_{b3}$, $x_{b4}$ & $NI$ \\
\hline
 & & 0.170 & & 0 & & - & -0.238, 3.571 & \\
 No & Bayesian & 0.057 & \textbf{0.227} & 0 & 0 & - & -, - & 0.245\\
\cline{2-9}
Rejection & Mutual- & 0.215 & & 0 & & - & -0.674, 4.007 & \\
  & Information & 0.024 & 0.239 & 0 & 0 & - & -, - & \textbf{0.260}\\
\hline
 & & 0.131 & & 0.083 & & 0.333 & -0.673, 0.162 & \\
 & Bayesian & 0.024 & \textbf{0.155} & 0.084 & 0.167 & 0.375 & 3.171, 4.006 & 0.285\\
\cline{2-9}
Rejection & Mutual- & 0.154 & & 0.118 & & 0.141 & -1.24, -0.0762 & \\
  & Information & 0.006 & 0.160 & 0.068 & 0.186 & 0.445 & 3.409, 4.571 & \textbf{0.297}\\
\hline
\end{tabular}
\end{table*}

Following the similar analysis above, one can reach a consistent observation for conducting 
a parametric study on $\sigma_1/\sigma_2$ in binary classifications. When two classes are 
well balanced, that is, $p(t_1)=p(t_2)$, both types of classifiers will produce larger errors in association with the 
larger-variance class. However, mutual-information classifiers always add more cost weight 
on the misclassification from a smaller-variance class. In other words, mutual-information 
classifiers prefer to generate a smaller error on a smaller-variance class in comparison 
with Bayesian classifiers when using zero-one cost functions (Table IV). This performance 
behavior seems closer to our intuitions in binary classifications under the condition 
of a balanced class dataset. When two classes are significantly different from their associated 
variances, a smaller-variance class generally represents an interested signal embedded within 
noise which often has a larger variance. The common practices in such classification scenarios 
require a larger cost weight on the misclassification from a smaller-variance class, and vice 
verse from a larger-variance class.

If a reject option is enforced for the following data:
\begin{equation}
 \begin{array}{r@{}l}
& Rejection: \\
& \mu_1=-1,\sigma_1=2,p(t_1)=0.5,\lambda_{11}=0,\lambda_{12}=1.2,\\ \nonumber
&\lambda_{13}=0.2,\\
& \mu_2=1,\sigma_2=1,p(t_2)=0.5,\lambda_{21}=1,\lambda_{22}=0,\\
&\lambda_{23}=0.6
\end{array}
\end{equation}
four boundary points are required to determine classification regions as shown in Fig 1d. 
For the given cost terms, a Bayesian classifier shows a lower error rate and a lower reject 
rate. While the rejects are almost equal between two classes, the errors are significantly 
different. One is able to adjust the errors and rejects by changing cost terms. For 
mutual-information classifiers, however, a balance is automatically made among error types 
and reject types. The results, shown in Table IV, are considered for carrying the feature of 
objectivity in evaluations since no cost terms are specified subjectively. Note that a reject 
option enables both classifiers to reach higher values on their $NI$'s than those in the case 
of without rejection. Because no ``{\it one-to-one}'' relations exist among the thresholds 
and the cost terms in a rejection case, one will fail to acquire a unique set of the equivalent cost 
terms between the Bayesian classifier and the mutual information classifier. For example, 
two sets of cost terms below will produce the same Bayesian classifiers based on the
given solutions of the mutual information classifier:
\begin{equation}
\left\{ \begin{array}{r@{}l}
& \lambda_{11}=0,~\lambda_{12}=1,~\lambda_{13}=0.0376,\\ \nonumber
& \lambda_{21}=1,~\lambda_{22}=0,~\lambda_{23}=0.772
\end{array}\right.
\end{equation}
or
\begin{equation}
\left\{ \begin{array}{r@{}l}
& \lambda_{11}=0,~\lambda_{12}=2.247,~\lambda_{13}=1,\\ \nonumber
& \lambda_{21}=7.069,~\lambda_{22}=0,~\lambda_{23}=1.
\end{array}\right.
\end{equation}
The meanings for two sets of cost terms are different. The first set indicates 
the same costs for errors, but the second one suggests the same costs for
rejects. The results above imply an intrinsic 
problem of ``{\it non-consistency}'' for interpreting cost terms. One needs to be cautious 
about this problem when setting cost terms to Bayesian classifiers. This phenomenon occurs 
only in the case that a reject option is considered, but does not in the case without rejection. 
If the knowledge about thresholds exists, abstaining classifiers are better to apply $T_{rk}$ 
directly for the input data (Table III), instead of employing cost terms. If no information is given
about the thresholds or cost terms, mutual-information classifiers are able to provide 
an objective, or initial, reference of $T_{rk}$ for Bayesian classifiers in cost sensitive learning. 

\begin{example}
{\it One cross-over point}. The given inputs in this example are:
\begin{equation}
 \begin{array}{r@{}l}
& No~rejection: \\
& \mu_1=-1,~\sigma_1=1,~\lambda_{11}=0,~\lambda_{12}=1,\\ \nonumber
& \mu_2=1,~\sigma_2=1,~\lambda_{21}=1,~\lambda_{22}=0,\\
& p(t_1)=0.5,2/3,0.8,0.9,0.99,0.999,0.9999\\
& p(t_2)=0.5,1/3,0.2,0.1,0.01,0.001,0.0001
\end{array}
\end{equation}
\end{example}

Specific attention is paid to the class imbalanced data. When Class 2 alters 
from ``{\it balanced}'', ``{\it minority}'' to ``{\it rare}'' status in the whole data, 
we need to find out what behaviors both types of classifiers will display. For this purpose, 
a natural scheme with zero-one cost terms is set for Bayesian classifiers. Numerical 
investigations are conducted in this example. Table V lists the results of classifiers 
on the given data. If following the conventional term $FNR$ for ``{\it false negative rate}'' 
in binary classifications, which is defined as:
\begin{equation}
 \label{equ:51}
FNR = \dfrac{E_2}{p(t_2)}
\end{equation}
one can examine behaviors of $FNR$ with respect to the ratio $p(t_1)/p(t_2)$. Sometimes, 
$FNR$ is also called a ``{\it miss rate}'' \cite{duda2001}. Two types of classifiers show 
the same results when two classes are exactly balanced, that is, 
$p(t_1)/p(t_2)=1$. A single boundary point (Fig. 1a) separates two classes at the exact 
cross-over point ($x_b=x_c=0$). When one class, say $p(t_2)$ for Class 2, becomes smaller, 
the boundary point of Bayesian classifier moves toward to the mean point $(\mu_2=1)$ of 
Class 2 (as pointed out in [\cite{duda2001}, page 39]), and passes it finally. For keeping 
the smallest error, a Bayesian classifier will sacrifice the minority class. The results in
Table V confirm Theorem 3 numerically on the Bayesian classifiers.
Fig. 3 shows such behavior from the plot of ``$E_2/p(t_2)$ vs. $p(t_1)/p(t_2)$''. Note that 
the plots for the range from $10^{-4}$ to $10^0$ on the $p(t_1)/p(t_2)$ axis are also depicted 
based on the data in Table V. For example, at the data point of $p(t_1)/p(t_2)=1/2$, one can 
get $E_2/p(t_2)=0.0594/(2/3)$, where 0.0594 is taken from $E_1$ for the data at 
$p(t_1)/p(t_2)=2$. The response of $E_2/p(t_2)$, representing the false negative rate, 
shows a distinguished property of Bayesian classifiers. One can observe that the 
complete set of Class 2 could be misclassified when it becomes extremely rare. This 
finding explains another reason for the question: ``{\it Why do classifiers perform 
worse on the minority class?}'' in \cite{weiss2001}. 

\begin{table*}
 \label{tab:5}
\renewcommand{\arraystretch}{1.3}
\caption{Results of Example 2 on Univariate Gaussian Distributions}
\setlength{\tabcolsep}{5pt}
\centering
\begin{tabular}{|c|c|c|c|c|c|c|c|c|}
 \hline
 Classifier & $p(t_1)/p(t_2)$ & 1 & 2 & 4 & 9 & 99 & 999 & 9999 \\
 Type & $[p(t_1), p(t_2)]$ & $[0.5, 0.5]$ & $[2/3, 1/3]$ & $[0.8, 0.2]$ 
& $[0.9, 0.1]$ & $[0.99, 0.01]$ & $[0.999, 0.001]$ & $[0.9999, 0.0001]$ \\
\hline
 & $E_1$ & 0.0793 & 0.0594 & 0.0362 & 0.0161 & 0.483e-3 & 0.422e-5 & 0.000 \\
 & $E_2$ & 0.0793 & 0.0856 & 0.0759 & 0.0539 & 0.903e-2 & 0.993e-3 & 0.1e-3\\
\cline{2-9}
Bayesian & $E_2/p(t_2)$ & 0.159 & 0.257 & 0.379 & 0.539 & 0.903 & 0.993 & 1.000 \\
\cline{2-9}
  & $x_b(=x_c)$ & 0.0 & 0.347 & 0.693 & 1.10 & 2.30 & 3.45 & 4.61 \\
\cline{2-9}
  & $H(T|Y)$ & 0.631 & 0.591 & 0.491 & 0.349 & 0.0756 & 0.0113 & 0.00147 \\
\cline{2-9}
  & $NI$ & 0.369 & 0.356 & 0.320 & 0.256 & 0.0644 & 0.00524 & 0.124e-3 \\
\hline
 & $E_1$ & 0.0793 & 0.0867 & 0.0852 & 0.0772 & 0.0585 & 0.0551 & 0.0547 \\
 & $E_2$ & 0.0793 & 0.0637 & 0.0451 & 0.0264 & 0.331e-2 & 0.343e-3 & 0.345e-4\\
\cline{2-9}
Mutual- & $E_2/p(t_2)$ & 0.159 & 0.191 & 0.225 & 0.264 & 0.331 & 0.343 & 0.345 \\
\cline{2-9}
 Information & $x_b$ & 0.0 & 0.126 & 0.246 & 0.367 & 0.562 & 0.597 & 0.601 \\
\cline{2-9}
& $H(T|Y)$ & 0.631 & 0.586 & 0.472 & 0.320 & 0.0629 & 0.00957 & 0.00129 \\
\cline{2-9}
  & $NI$ & 0.369 & 0.362 & 0.346 & 0.317 & 0.222 & 0.161 & 0.125 \\
\hline
\end{tabular}
\end{table*}

\begin{figure}[!t]
\centering
\includegraphics[width=3.in,height=2.in,bb= -15 0 260 180]{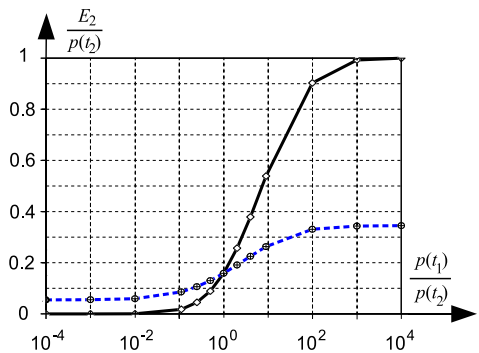}
\caption{Curves of ``$E_2/p(t_2)$ vs. $p(t_1)/p(t_2)$''. Solid curve: Bayesian classifier. 
Dashed curve: Mutual-information classifier.}
\label{fig:3}
\end{figure}

Mutual-information classifiers exhibit different behavior in the given dataset. 
The first important feature is that the boundary point will shift toward the mean 
point $(\mu_2=1)$ of Class 2 but will never go over it. The second feature informs 
that the response of $E_2/p(t_2)$ approaches asymptotically to a stable value, about 
0.345 in this example, for a large ratio of $p(t_1)/p(t_2)$. This feature indicates 
that mutual-information classifiers will never sacrifice a minority class completely 
in this specific example. A significant fraction of the rare class is identified correctly. 
Moreover, the curve of $E_2/p(t_2)$ also demonstrates a lower, yet non-zero, bound on error 
rate (about 0.054) when $p(t_1)/p(t_2)$ approaches to zero. This phenomenon implies that, 
for Gaussian distributions of classes, mutual-information classifiers generally do not hold 
a tendency of sacrificing a complete class in classifications. However, from a theoretical 
viewpoint, we still need to establish an analytical derivation of lower and upper bounds of 
$E_i/p(t_i)$ for mutual-information classifiers.

\begin{example}{\it Zero cross-over points.}
 The given data for two classes are:
\begin{equation}
  \begin{array}{r@{}l}
& \mu_1=0,~\sigma_1=2,~p(t_1)=0.8,\\ \nonumber
& \mu_2=0,~\sigma_2=1,~p(t_2)=0.2.
\end{array}
\end{equation}
\end{example}

Although no data are specified to the cost terms, it generally implies a zero-one lost 
function for them \cite{duda2001}. From eq. (\ref{equ:13}), one can see a case of zero 
cross-over point occurs in this example (Fig. 4c). For the zero-one setting to cost terms, the
Bayesian classifier will produce a specific classification result of ``{\it Majority-taking-all}'', 
that is, for all patterns identified as Class 1. The error gives to Class 2 only, and it holds 
the relation of $NI=0$, which indicates that no information is obtained from the classifier \cite{mackay2003}. 
One can imagine that the given example may describe a classification problem where a target class, with 
Gaussian distribution, is also corrupted with wider-band Gaussian noise in a frequency domain 
(Fig. 4a). The plots of $p(t_i)p(x|t_i)$ shows the overwhelming distribution of Class 1 over 
that of Class 2 (Fig. 4b). The plots on the posterior probability $p(t_i|x)$ indicate that 
Class 2 has no chance to be considered in the complete domain of $x$ (Fig. 4c).  

\begin{figure*}[!t]
\centering
\includegraphics[width=5.0in,height=1.7in,bb= 60 0 470 155]{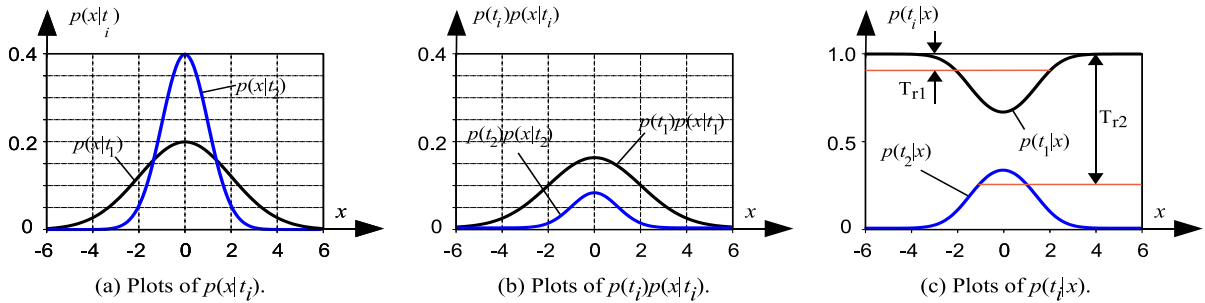}
\caption{Plots for Example 3 where (b)-(c) describe a signal (blue curve) embedded 
by wider-band noise (black curve).}
\label{fig:4}
\end{figure*}

Table VI lists the results for both types of classifiers. The Bayesian approach fails to achieve 
the meaningful results on the given data. When missing input data of $\lambda_{13}$ and $\lambda_{23}$, 
one cannot carry out the Bayesian approach for abstaining classifications. On the contrary, without specifying 
any cost term, mutual-information classifiers are able to detect the target class with a 
reasonable degree of accuracy. When no rejection is selected, less than two percentage error 
$(E2=1.53\%)$ happens to the target class. Although the total error $(E=51.4\%)$ is much higher 
than its Bayesian counterpart $(E=20\%, FNR=0\%)$, the result of about eight percentage point $(FNR=7.65\%)$ 
of the miss rate to the target is really meaningful in applications. If a reject option is engaged, 
the miss rate is further reduced to $FNR= 4.10\%$, but includes adding a reject rate of $29.1\%$ over 
total possible patterns. This example confirms again the unique feature of mutual-information 
classifiers. The results of $T_{rk}$ from mutual-information classifiers can also serve a useful 
reference for the design of Chow's abstaining classifiers, either with or without knowledge about 
cost terms. 

\begin{table*}
 \label{tab:6}
\renewcommand{\arraystretch}{1.3}
\caption{Results of Example 3 on Univariate Gaussian Distributions}
\setlength{\tabcolsep}{10pt}
\centering
\begin{tabular}{|c|c|c|c|c|c|c|c|c|}
 \hline
Reject & Classifier & $E_1$ & & $Rej_1$ & & $T_{r1}$ & $x_{b1}$, $x_{b2}$ & \\
Option & Type & $E_2$ & $E$ & $Rej_2$ & $Rej$ & $T_{r2}$ & $x_{b3}$, $x_{b4}$ & $NI$ \\
\hline
 & & 0.0 &  & 0 & & - & -, - & \\
No & Bayesian & 0.2 & 0.2 & 0 & 0 & - & -, - & 0.0\\
\cline{2-9}
Rejection & Mutual & 0.499 & & 0 & & - & -1.77, 1.77 & \\
  & Information & 0.0153 & 0.514 & 0 & 0 & - & -, - & 0.0803 \\
\hline
 & Mutual & 0.316 &  & 0.239 &  & 0.0945 & -2.04, -1.03 &  \\
Rejection & Information & 0.00819 & 0.324 & 0.0520 & 0.291 & 0.749 & 1.03, 2.04 & \textbf{0.0926}\\
\hline
\end{tabular}
\end{table*}

\subsection{Comparisons on Univariate Uniform Distributions}
\label{subsec:4.3}
Uniform distributions are very rare in classification problems. This 
section shows one example given from \cite{chow1970}. A specific effort is 
made on numerical comparisons between the two types of classifiers.

\begin{example} {\it Partially overlapping between two distributions.}
 The task for this example is to set the cost terms for controlling the 
decision results on the overlapping region for the given data from \cite{chow1970}:
\begin{equation}
 p(x|t_1) = \left\{ \begin{array}{r@{\quad}l}
1 & when ~ 0 \leq x \leq 1\\ \nonumber
0 & otherwise
\end{array}\right.
\end{equation}
\begin{equation}
 p(x|t_2) = \left\{ \begin{array}{r@{\quad}l}
1/2 & when ~ 0.5 \leq x \leq 2.5\\ \nonumber
0 & otherwise
\end{array}\right. 
\end{equation}
\begin{equation}
p(t_1)=p(t_2)=0.5. \nonumber
\end{equation}
In uniform distributions, a single independent parameter will be sufficient for classifications.  
Table VII lists the different results with respect to $T_r$.  
Note that the present results have extended Chow's 
abstaining classifiers by adding one more decision case of $f(x \in R_i)  = y_2$ than those in \cite{chow1970}. 
The extension is attributed to the three rules used in eq. (\ref{equ:9}), rather than two 
in Chow's classifiers, which demonstrates a more general solution for classifications. 
One can see that mutual-information classifiers will decide $f(x \in R_i)  = y_3$ from the 
given data of class distributions sine they receive the maximum value of $NI$. If no 
rejection is enforced, mutual-information classifiers will choose $f(x \in R_i)  = y_1$ for their 
solution.
\end{example}

\begin{table*}
 \label{tab:7}
\renewcommand{\arraystretch}{1.3}
\caption{Results of Example 4 on Univariate Uniform Distributions}
\setlength{\tabcolsep}{10pt}
\centering
\begin{tabular}{|c|c|c|c|c|c|c|}
 \hline
$T_{r}$ & Decision on $R_i$& $E_1$, $E_2$ & $E$ & $Rej_1$, $Rej_2$ & $Rej$ &  $NI$\\
\hline
$1/3<T_r<2/3$ & $f(x \in R_i)  = y_1$ & 0.0, 0.125 & 0.125 & 0, 0 & 0 &  0.549 \\
\hline
$2/3 \leq T_r \leq 1$ & $f(x \in R_i)  = y_2$ & 0.250, 0 & 0.250 & 0, 0 & 0 &  0.311 \\
\hline
$0 \leq T_r \leq 1/3$ & $f(x \in R_i)  = y_3$ & 0, 0 & 0 & 0.250, 0.125 & 0.375 &  \textbf{0.656} \\
\hline
\end{tabular}
\end{table*}

\section{Conclusions}
\label{sec:5}
This work explored differences between Bayesian classifiers and mutual-information classifiers. 
Based on Chow's pioneering work \cite{chow1957}\cite{chow1970}, the author revisited Bayesian 
classifiers on two general scenarios for the reason of their increasing popularity 
in classifications. The first was on the zero-one cost functions for classifications  
without rejection. The second was on the cost distinctions among error types and reject 
types for abstaining classifications. In addition, the paper focused on the analytical study of mutual-information classifiers 
in comparison with Bayesian classifiers, which showed a basis for novel design or analysis of 
classifiers based on the entropy principle. The general decision rules were derived for both Bayesian 
and mutual-information classifiers based on the given assumptions. Two specific theorems were 
derived for revealing the intrinsic problems of Bayesian classifiers in applications under the 
two scenarios. One theorem described that Bayesian classifiers have a tendency of overlooking the 
misclassification error which is associated with a minority class. This tendency will degenerate 
a binary classification into a single class problem for the meaningless solutions. The other theorem 
discovered the parameter redundancy of cost terms in abstaining classifications. This weakness is 
not only on reaching an inconsistent interpretation to cost terms. The pivotal difficulty will be 
on holding the objectivity of cost terms. In real applications, information about cost terms is 
rarely available. This is particularly true for reject types. While Berger explained the demands 
for ``{\it objective Bayesian analysis}'' \cite{berger2006}, we need to recognize that this goal 
may fail from applying cost terms in classifications. In comparison, mutual-information 
classifiers do not suffer such difficulties. Their advantages without requiring cost terms 
will enable the current classifiers to process abstaining classifications, like
a new folder of ``{\it Suspected Mail}'' in Spam filtering \cite{sahami1998}. Several numerical 
examples in this work supported the unique benefits of using mutual-information classifiers in 
special cases. The comparative study in this work was not meant to replace Bayesian classifiers 
by mutual-information classifiers. Bayesian and mutual-information classifiers can form 
``{\it complementary rather than competitive} (words from Zadeh \cite{zadeh1995})'' 
solutions to classification problems. 
However, this work was intended to highlight their differences from theoretical 
studies. More detailed discussions to the differences between the two types of classifiers were 
given in Section IV. As a final conclusion, a simple answer to the question title is summarized 
below: 

``{\it Bayesian and mutual-information classifiers are different essentially from their 
applied learning targets. From application viewpoints, Bayesian classifiers are more suitable 
to the cases when cost terms are exactly known for trade-off of error types and reject types. 
Mutual-information classifiers are capable of objectively balancing error types and reject types 
automatically without employing cost terms, even in the cases of extremely class-imbalanced datasets, 
which may describe a theoretical interpretation why humans are more concerned about the accuracy 
of rare classes in classifications}''.

\appendices
\section{Proof of Theorem 1}
\label{app:proof1}
\begin{proof}
The decision rule of Bayesian classifiers for the ``{\it no rejection}'' case is well known in \cite{duda2001}. 
Then, only the rule for the ``{\it rejection}'' case is studied in the present proof. Considering eq. (\ref{equ:6a}) 
first from (\ref{equ:5a}), a pattern $\textbf{x}$ is decided by a Bayesian classifier to be 
$y_1$ if $risk (y_1|\textbf{x}) < risk (y_2|\textbf{x})$ and $risk (y_1|\textbf{x}) < risk (y_3|\textbf{x})$. Substituting 
eqs. (\ref{equ:1}) and (\ref{equ:2}) into these inequality equations will result to:
\begin{equation}\tag{A1}
\label{equ:A1}
 \begin{array}{r@{~}l}
Decide~y_1~if & \dfrac{p(\textbf{x}|t_1)p(t_1)}{p(\textbf{x}|t_2)p(t_2)} > \dfrac{\lambda_{21}-\lambda_{22}}{\lambda_{12}-\lambda_{11}}\\
and ~ & \dfrac{p(\textbf{x}|t_1)p(t_1)}{p(\textbf{x}|t_2)p(t_2)} > \dfrac{\lambda_{21}-\lambda_{23}}{\lambda_{13}-\lambda_{11}}.
\end{array}
\end{equation}
Similarly, one can obtain 
\begin{equation}\tag{A2}
\label{equ:A2}
 \begin{array}{r@{~}l}
Decide~y_2~if & \dfrac{p(\textbf{x}|t_1)p(t_1)}{p(\textbf{x}|t_2)p(t_2)} \leq \dfrac{\lambda_{21}-\lambda_{22}}{\lambda_{12}-\lambda_{11}}\\
and ~ & \dfrac{p(\textbf{x}|t_1)p(t_1)}{p(\textbf{x}|t_2)p(t_2)} \leq \dfrac{\lambda_{23}-\lambda_{22}}{\lambda_{12}-\lambda_{13}},
\end{array}
\end{equation}
and eq. (\ref{equ:6c}) respectively. Eq. (\ref{equ:A1}) describes that a single upper bound within two 
boundaries will control a pattern $\textbf{x}$ to be $y_1$. Similarly, eq. (\ref{equ:A2}) describes a lower 
bound for a pattern $\textbf{x}$ to be $y_2$. From the constraints (\ref{equ:3}), one cannot determine which 
boundaries will be upper bound or lower bound. However, one can determine them from the following two 
hints in classifications:
\begin{enumerate}
 \item[A.] Eq. (\ref{equ:6c}) describes a single lower boundary and a single upper boundary for a 
pattern $\textbf{x}$ to be $y_3$. 
 \item[B.] The upper bound in (\ref{equ:A1}) and the lower bound in (\ref{equ:A2}) should be 
coincident with one of the boundaries in (\ref{equ:6c}) respectively so that classification 
regions from $R_1$ to $R_3$ will cover a complete domain of the pattern $x$ (see Fig. 1c-d).    
\end{enumerate}

The hints above suggest the novel constraints for $\lambda_{ij}$ as shown in eq. (\ref{equ:6d}). 
Any violation of the constraints will introduce a new classification region $R_4$, which is not 
correct for the present classification background. The constraints of thresholds for rejection 
(\ref{equ:6e}) can be derived directly from (\ref{equ:6c}) and (\ref{equ:6d}).  
\end{proof}

\section{Tighter Bounds between Conditional Entropy and 
Bayesian Error in Binary Classifications}
In the study of relations between mutual information ($I$)
and Bayesian error ($E$), two important studies are reported on 
the lower bound ($LB$) by Fano \cite{Fano1961} and the upper bound ($UB$) 
by Kovalevskij \cite{Kovalevskij1965} in the forms of

\begin{equation}\tag{B1}
\label{equ:B1}
LB: ~  E \geq \dfrac{H(T)-I(T,Y)-H(E)}{log_2 (m-1)}=\dfrac{H(T|Y)-H(E)}{log_2 (m-1)},
\end{equation}
\begin{equation}\tag{B2}
\label{equ:B2}
UB: ~  E \leq \dfrac{H(T)-I(T,Y)}{2}=\dfrac{H(T|Y)}{2} ,~~~~~~~~~~~~~~~~~~~~
\end{equation}
where $m$ is the total number of classes in $T$, $H(E)$ is 
the binary Shannon entropy, and $H(T|Y)$ is 
called conditional entropy which can be derived from a general relation \cite{duda2001}: 
\begin{equation}\tag{B3}
\label{equ:B3}
I(T,Y)=I(Y,T)=H(T)-H(T|Y)=H(Y)-H(Y|T).
\end{equation}

For binary classifications $(m=2)$, a tighter Fano's bound in \cite{Vajda2007}
\cite{Fisher2009} is adopted. Based on the rationals of Bayesian error, we
suggest the tighter upper and  lower bounds in the forms of:
\begin{equation}\tag{B4}
\label{equ:B4}
Modified ~LB: ~ H(E) \geq H(T|Y),~and ~ 0 \leq E, ~~
\end{equation}
\begin{equation}\tag{B5}
\label{equ:B5}
Modified ~UB:  ~ E \leq min (p(t_1), p(t_2), \dfrac{H(T|Y)}{2}).
\end{equation}

Fig. 5 shows the bounds in binary classifications, 
which is different from ``$ I(T,Y) ~vs. ~E$'' plots in \cite{Fisher2009}.
Because of the equivalent relations \cite{hu2008}:
\begin{equation}\tag{B6}
\label{equ:B6}
max ~I(T,Y)= min ~H(T|Y),
\end{equation}
the plots for $H(T|Y)$  is preferable, which does not require the information of  
$H(T)$.
One is able to draw the lower-bound curve from (B4), but unable to
show its explicit form for $E$.  
The areal feature of the enclosed bounds suggests two important properties about the
relations. The first is due to the approximations in the derivations 
of the bounds \cite{Fano1961} \cite{Kovalevskij1965}.   
The second represents an intrinsic property of  no ``{\it one-to-one}'' 
relations between mutual information and accuracy 
in classifications \cite{wang2008}. 

Triangles and circles shown in Fig. 5 represent the paired data 
in Table V from Bayesian classifiers
and mutual information classifiers, respectively.
They clearly demonstrate the specific forms in their positions within the same pairs.
The circle position is either coincident or ``{\it up and/or left}'' to its counterpart.  
These forms are attributed to the different directions of driving force  
for two types of classifiers. One is for ``$min ~ E$'' and the other for ``$min ~ H(T|Y)$''.

Important findings are observed in related to the bounds. First, the triangles
demonstrate Fano's bound in eq. (B4) to be a very tight lower bound. 
Second, an upper bound of $E_{max}$ exists according to Theorem 3, which is
tighter than a constant one ($=0.5$) in \cite{Vajda2007}. 
When $p_{min}$ decreases as shown 
in Table V, the upper bound from the maximum Bayesian error will become
closer to its associated data. Third, the Fano's lower bound is effective
for all classifiers, including mutual information classifiers. However, 
the upper bounds, even the constant one ($=0.5$) becomes invalid for mutual
information classifiers (see the data $E=0.514$ in Table VI). 

The observations above indicate the necessity of further investigation into the upper bounds
for better descriptions of the relations. If much tighter upper bounds are possible, they 
are desirable to disclose their theoretical insights between the two types of classifiers. 

\begin{figure}
\centering
\includegraphics[width=3.4in,height=2.1in,bb= -15 0 340 220]{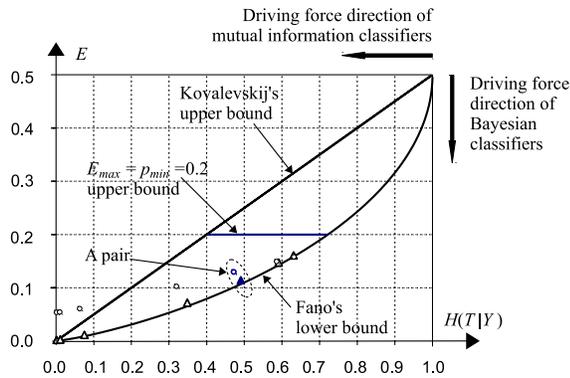}
\caption{The bounds between conditional entropy $H(T|Y)$ 
and Bayesian error in binary classifications. 
Triangles and circles are the data in Table V from Bayesian classifiers
and mutual information classifiers, respectively.
An upper bound from the maximum Bayesian error exists,
say, $E_{max}=0.2$ for the filled triangle.}
\label{fig:B1}
\end{figure}

\section*{Acknowledgments}
This work was supported in part by NSFC $\#$61075051. The assistances from Mr. Yajun Qu and Mr. Christian Ocier in
the text preparation are gratefully acknowledged.

%

\bibliographystyle{IEEETrans}


%





\end{document}